\documentclass[conference,compsoc]{IEEEtran}
\IEEEoverridecommandlockouts
\usepackage{amsmath,amsfonts,nccmath}
\usepackage{algorithmic}
\usepackage[linesnumbered,ruled,noend]{algorithm2e}

\usepackage{array}
\usepackage{textcomp}
\usepackage{stfloats}
\usepackage{url}
\usepackage{verbatim}
\usepackage{graphicx}
\usepackage{cite}
\usepackage{enumitem}
\usepackage{multirow}
\usepackage{soul}
\usepackage{xcolor}
\usepackage{color}
\usepackage{tikz}
\usepackage{amsmath}
\usepackage{url}
\usepackage{cite}
\usepackage{graphicx}
\usepackage{booktabs,threeparttable}
\usepackage{comment}
\usepackage{cuted}
\usepackage{pifont}
\usepackage{caption}
\usepackage{subcaption}
\usepackage{makecell}
\usepackage{color, colortbl}
\usepackage{amsthm}
\usepackage{soul}
\usepackage{marvosym}

\definecolor{Gray}{gray}{0.85}
\definecolor{verylightgray}{rgb}{.97,.97,.97}
\definecolor{mygreen}{RGB}{24,141,31}
\definecolor{myred}{RGB}{142,0,8}
\definecolor{mypurple}{RGB}{107,29,111}

\usepackage{adjustbox}
\usepackage{hyperref}
\hypersetup{
    colorlinks=true,
    linkcolor=mygreen,
    filecolor=mygreen,      
    urlcolor=mygreen,
    citecolor=mypurple,
}

\newcommand{\hzy}[1]{\textcolor{green}{hzy: #1}}

\newtheorem{lemma}{Lemma}

\newcommand{\mpt}{\emph{MPT}}
\newcommand*\circlednum[1]{\raisebox{.4pt}{\textcircled{\raisebox{-.9pt} {#1}}}}
\definecolor{shitgreen}{RGB}{80,126,50}

\newcommand{\colorboxred}[1]{%
  \tikz[baseline=(char.base)]{
    \node[rectangle, fill=yellow, draw=shitgreen, inner sep=2pt, outer sep=0pt] (char) {\textcolor{red}{#1}};
  }%
}

\definecolor{ndssblue}{RGB}{133, 165, 195}
\newcommand{\colorboxblue}[1]{%
  \tikz[baseline=(char.base)]{
    \node[rectangle, fill=ndssblue, inner sep=2pt, outer sep=0pt] (char) {\textcolor{black}{#1}};
  }%
}

\usepackage{listings}
\newcommand{\ignore}[1]{}

\definecolor{verylightgray}{rgb}{.97,.97,.97}

\lstdefinelanguage{Solidity}{
	keywords=[1]{emit, anonymous, assembly, assert, balance, break, call, callcode, case, catch, class, constant, continue, contract, debugger, default, delegatecall, delete, do, else, event, export, external, false, finally, for, function, gas, if, implements, import, in, indexed, instanceof, interface, internal, is, length, library, log0, log1, log2, log3, log4, memory, modifier, new, payable, pragma, private, protected, public, pure, push, require, return, returns, revert, selfdestruct, send, storage, struct, suicide, super, switch, then, this, throw, transfer, true, try, typeof, using, value, view, while, with, addmod, ecrecover, keccak256, mulmod, ripemd160, sha256, sha3}, %
	keywordstyle=[1]\color{blue}\bfseries,
	keywords=[2]{address, bool, byte, bytes, bytes1, bytes2, bytes3, bytes4, bytes5, bytes6, bytes7, bytes8, bytes9, bytes10, bytes11, bytes12, bytes13, bytes14, bytes15, bytes16, bytes17, bytes18, bytes19, bytes20, bytes21, bytes22, bytes23, bytes24, bytes25, bytes26, bytes27, bytes28, bytes29, bytes30, bytes31, bytes32, enum, int, int8, int16, int24, int32, int40, int48, int56, int64, int72, int80, int88, int96, int104, int112, int120, int128, int136, int144, int152, int160, int168, int176, int184, int192, int200, int208, int216, int224, int232, int240, int248, int256, mapping, string, uint, uint8, uint16, uint24, uint32, uint40, uint48, uint56, uint64, uint72, uint80, uint88, uint96, uint104, uint112, uint120, uint128, uint136, uint144, uint152, uint160, uint168, uint176, uint184, uint192, uint200, uint208, uint216, uint224, uint232, uint240, uint248, uint256, var, void, ether, finney, szabo, wei, days, hours, minutes, seconds, weeks, years},	%
	keywordstyle=[2]\color{teal}\bfseries, 
	keywords=[3]{block, blockhash, coinbase, difficulty, gaslimit, number, timestamp, msg, data, gas, sender, sig, value, now, tx, gasprice, origin},	%
	keywordstyle=[3]\color{violet}\bfseries,
	keywords=[4]{onlyOwner},
	keywordstyle=[4]\color{red}\bfseries,
	identifierstyle=\color{black},
	sensitive=false,
	comment=[l]{//},
	morecomment=[s]{/*}{*/},
	commentstyle=\color{gray}\ttfamily,
	stringstyle=\color{red}\ttfamily,
	morestring=[b]',
	morestring=[b]"
}

\newcommand{\chain}{Ethereum}

\newcommand{\dosattack}{\textsc{Nurgle}}

\newcommand{\dosformula}[1]{\textsl{\footnotesize #1}}
\newcommand{\dossmallformula}[1]{\textsl{\scriptsize #1}}
\newcommand{\doscode}[1]{\texttt{\footnotesize #1}}
\newcommand{\dostinyformula}[1]{\textsl{\tiny #1}}

\def\BibTeX{{\rm B\kern-.05em{\sc i\kern-.025em b}\kern-.08em
    T\kern-.1667em\lower.7ex\hbox{E}\kern-.125emX}}
\begin{document}

\title{\dosattack{}: Exacerbating Resource Consumption in Blockchain State Storage via MPT Manipulation}

\author{
\IEEEauthorblockN{
Zheyuan He$^\dag$$^\P$\thanks{$^\P$Co-first authors.}, 
Zihao Li$^\ddag$$^\P$, 
Ao Qiao$^\dag$, 
Xiapu Luo$^\ddag$\textsuperscript{\Letter}\thanks{\textsuperscript{\Letter} Corresponding authors.},
Xiaosong Zhang$^\dag$\textsuperscript{\Letter},\\
Ting Chen$^\dag$\textsuperscript{\Letter},
Shuwei Song$^\dag$,
Dijun Liu$^\S$, and
Weina Niu$^\dag$
}

\IEEEauthorblockA{$^\dag$University of Electronic Science and Technology of China\quad $^\ddag$The Hong Kong Polytechnic University\quad $^\S$Ant Group}
}

\maketitle

\begin{abstract}
Blockchains, with intricate architectures, encompass various components, e.g., consensus network, smart contracts, decentralized applications, and auxiliary services. 
While offering numerous advantages, these components expose various attack surfaces, leading to severe threats to blockchains. 
In this study, we unveil a novel attack surface, i.e., the state storage, in blockchains. The state storage, based on the Merkle Patricia Trie, plays a crucial role in maintaining blockchain state.
Besides, we design \dosattack{}, the first Denial-of-Service attack targeting the state storage.
By proliferating intermediate nodes within the state storage, \dosattack{} forces blockchains to expend additional resources on state maintenance and verification, impairing their performance.
We conduct a comprehensive and systematic evaluation of \dosattack{}, including the factors affecting it, its impact on blockchains, its financial cost, and practically demonstrating the resulting damage to blockchains. 
The implications of \dosattack{} extend beyond the performance degradation of blockchains, potentially reducing trust in them and the value of their cryptocurrencies.
Additionally, we further discuss three feasible mitigations against \dosattack{}.
At the time of writing, the vulnerability exploited by \dosattack{} has been confirmed by six mainstream blockchains, and we received thousands of USD bounty from them.
\end{abstract}

\section{Introduction}
\label{sec_intro}

Blockchains, with intricate architecture, have evolved various components, e.g., 
consensus network,
smart contracts,
decentralized applications,
and auxiliary services~\cite{zhou2022sok}. These components expose various attack surfaces, leading to severe threats to blockchains~\cite{zhou2022sok}. Among these threats, the frequency and severity of Denial-of-Service (DoS) attacks have been rising~\cite{li2021deter}.
DoS attacks deny the service of corresponding components, and compromise the operations of blockchain.
For instance, the DoS incident~\cite{DosIncident1,fluffyosdi} on the consensus network induces a hard fork and abandons over 30 blocks (worth 8.6M USD) atop the Ethereum~\cite{wood2014ethereum}. 

\ignore{
\noindent\textbf{Existing works.} Academic and industrial communities continuously explore new attack surfaces of blockchain under DoS threats, involving four attack surfaces as follows. i) Consensus network~\cite{fluffyosdi, heopartitioning, prunster2022total,tran2020stealthier,saad2021syncattack, chen2022tyr, heilman2015eclipse,wang2022sdos}.
Based on the fault-tolerant mechanism, consensus network
assists blockchain nodes in achieving an agreement on the latest state. Compromising consensus network will incur the blockchain to violate its consensus functionalities~\cite{chen2022tyr}. As an example, Heo et al.~\cite{heopartitioning} isolate peers in the consensus network by occupying connections with the peers, and finally succeed in paralyzing the consensus network. ii) Txpools~\cite{li2021deter,saad2018poster,wu2020survive}. Txpools maintain pending transactions from blockchain users, and miners/validators will pack transactions from their txpools to blockchain.
Adversaries sway the security of blockchain by interfering with the transaction packing involved in txpools. For instance, DETER~\cite{li2021deter} exploits how Ethereum txpools handle pending transactions to disable txpools, and deny blockchain services. 
iii) Auxiliary services (AUX)~\cite{zhou2022sok, NguyenTc2023, li2021strong,yaish2023speculative}. 
Auxiliary services refer to entities facilitating blockchain's efficiency at off-chain.
DoS attackers can refuse users to utilize services provided by AUX.
For example, 
Li et al.~\cite{li2021strong} 
propose a DoS attack to paralyze RPC services, which causes users to be unable to access Ethereum.
iv) Smart contracts~\cite{ndss_broken_metro,chen2017adaptive,ghaleb2022etainter,grech2020madmax,suicide_dos,extcodesize_dos}. Smart contracts are automation programs executed in blockchain. One of their security issues 
is
under-priced opcodes~\cite{chen2017adaptive}, whose gas cost is substantially less than corresponding consumed resources. 
Attackers can lead blockchain to consume extremely high resources while executing under-priced opcodes.
For example, Broken Metre~\cite{ndss_broken_metro} constructs contracts consisting of under-priced opcodes to undermine Ethereum's efficiency.
}

 \begin{figure}[]
 
\small
	\centering
	\includegraphics[width=0.99\linewidth]{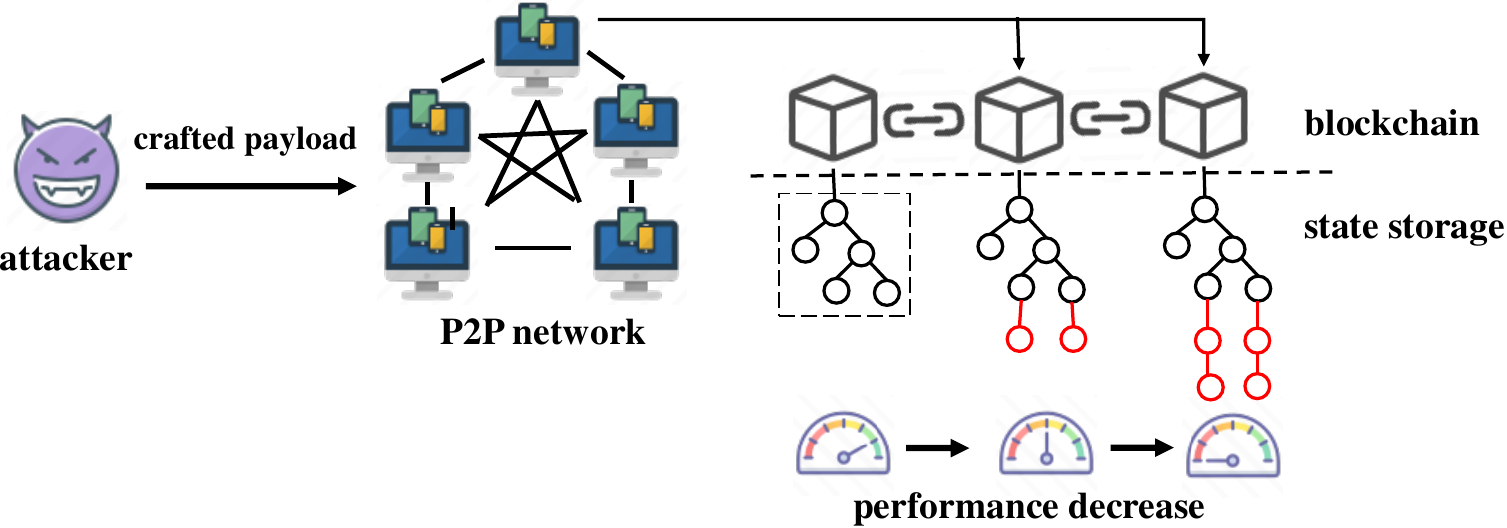}
	\caption{The workflow of \dosattack{}. The attacker first submits the crafted payload by transactions to the P2P network. Once the transactions are included in the blockchain, they can manipulate the state storage of the MPT structure, and finally impair the performance of the blockchain.}
	\label{fig_workflow}
\end{figure}

Academic communities continuously explore new attack surfaces of blockchain under DoS threats, involving four attack surfaces, i.e., consensus network, txpools, auxiliary services, and smart contracts (cf. details in \S\ref{sec_related}).
However, despite serving as the major performance
bottleneck of blockchain~\cite{Li2023LVMT,ponnapalli2021rainblock,ethanos}, DoS security concerning state storage of blockchain has never been explored.

To our best knowledge, 
we first reveal a new DoS attack surface of blockchain, i.e., state storage, which is used to maintain the blockchain state~\cite{offchain,ethanos}.
Blockchain state is the persistent data of the blockchain (e.g., account balance and contract 
persistent variables~\cite{offchain,bose2022sailfish}).
Besides, blockchain state is managed by a storage structure, typically Merkle Patricia Trie (MPT)~\cite{mptanalysis}. 
Hence, modifying the blockchain state consumes numerous resources.
Specifically,
the state storage of MPT structure (termed by \mpt{}) utilizes a tree structure to manage blockchain state. 
A leaf node in \mpt{} stores the value of persistent data (e.g., account balance), and all intermediate nodes in the path  
from the root node to the leaf node correspond to the key of the data's value
(e.g., account address)~\cite{Li2023LVMT}.
As the first systematic study on the security of blockchain state storage,
our study sheds light on the importance of securing the state storage,
and facilitates researchers and developers to propose more robust designs.

\noindent\textbf{Threat goals.} In our study, we propose \dosattack{}, a novel DoS attack towards blockchain state storage.
\dosattack{} aims to cripple blockchain's performance by raising time cost in interacting with its state storage. The design of \dosattack{} is inspired by our two observations:

\noindent
$\bullet$
Heavy burden of state maintenance (\S\ref{sec_observation_stateop}).
We categorize four classes of time-consuming blockchain operations (Table~\ref{tab_stateops}).
Besides,
we unevil that
the time cost of the operations significantly relies on the number of \mpt{} nodes involved within the state storage during blockchain execution.

\noindent
$\bullet$
Flaw of gas mechanism (\S\ref{sec_observation_gasflaw}).
Gas mechanism~\cite{gas_intro} is responsible for determining the financial cost incurred by users when accessing blockchain's resources.
However, this mechanism fails to consider intermediate nodes in state storage, which are updated when the state is modified (Fig.~\ref{gas_incon_demo}).

\ignore{
It reports that
81\% execution time of blockchain costs in interacting with
state storage~\cite{Li2023LVMT,ponnapalli2021rainblock,ethanos}. 
We further uncover that the time cost in interacting with state storage linearly increases with the number of involved nodes in state storage, leading the performance of blockchain heavily depend on the number of nodes in state storage.
We categorize four classes of operations in interacting with state storage in Table~\ref{tab_stateops}.
Our first observation motivates \dosattack{} to involve more nodes associated with the four classes of operations in state storage, resulting in degrading blockchain's execution performance and exacerbating consumed resources.

\noindent
$\bullet$
Flaw of gas mechanism (\S\ref{sec_observation_gasflaw}). Gas mechanism~\cite{gas_intro} deters adversaries from introducing infinite resource consumption to exploit blockchain by the cost of gas~\cite{gas_intro}. However, we reveal that the current gas mechanism fails to accurately reflect the actual consumed resource during the state modification in state storage.
Concretely, gas mechanism does not take into account intermediate nodes in state storage, which are updated when the state is modified (Fig.~\ref{gas_incon_demo}).
Our second observation inspires \dosattack{} to proliferate and involve more intermediate nodes with costing less gas fee, leading to degrading the execution performance of blockchain and exacerbating resource consumption.
}

According to
our two observations, \dosattack{} can proliferate intermediate nodes with few cost to increase extra resource consumption for maintaining the intermediate nodes, thereby degrading overall blockchain performance.
Fig.~\ref{fig_workflow} displays \dosattack{}'s workflow.
Specifically, \dosattack{} first constructs data and submits the crafted payload to the P2P network. Once the payload is included in blockchain, \dosattack{} can manipulate \mpt{}, and finally impair the performance of the blockchain.
Furthermore, as the manipulated \mpt{} persists on the blockchain, \dosattack{} persistently threatens the blockchain in all subsequent blocks. 

\ignore{The threat model of \dosattack{} contains an adversary with limited computing resources and analysis capabilities towards the storage structure of blockchain state and a victim blockchain with the state storage maintained in MPT structure (cf. details in \S \ref{sec_model}). }

\noindent\textbf{Attack scope.} 
Blockchains, adopting MPT~\cite{mptanalysis} for handling state storage~\cite{ruan2021blockchains,wei2022survey,Li2023LVMT}, are under threats of \dosattack{}.
In practice, MPT structure is widely applied in mainstream blockchain platforms, e.g., Ethereum~\cite{wood2014ethereum}, Binance Smart Chain (BSC)~\cite{bnb}, Polygon~\cite{polygon}, Avalanche~\cite{Avalanche}, Optimism~\cite{Optimism}, and Polkadot~\cite{wood2016polkadot}.
To our best knowledge, 588 blockchain platforms (cf. lists in Appendix \ref{sec_scope}) are under threats of \dosattack{}, which are compatible with Ethereum ecosystem~\cite{Ethereum_eco,yi2022blockscope}. In addition, other 153 blockchain platforms (cf. lists in Appendix \ref{sec_scope}) compatible with Polkadot ecology are also affected by \dosattack{}~\cite{Polkadot_eco} (\S \ref{sec_discuss}). 

\noindent\textbf{Attack design.} 
\dosattack{} strategically manipulates \mpt{} to impair blockchain performance. 
Specifically, \dosattack{} aims to raise \mpt{}'s depth by proliferating intermediate nodes, i.e., nodes between the root node and leaf nodes. When \mpt{}'s depth increases, 
blockchain consumes more resources to maintain expanded intermediate nodes in \mpt{}.

Within \mpt{},
which is organized as a prefix tree~\cite{mptanalysis}, each node is located by a unique value (denoted as \dosformula{indexing}, e.g., \doscode{acd3f} in Fig.~\ref{eth_stat}).
Besides, all \mpt{} nodes are ordered by the prefixes of their \dosformula{indexing}~\cite{mptanalysis} (\S\ref{sec_background}).
Hence, to proliferate intermediate nodes in a specific path of \mpt{},
\dosattack{} needs to craft a leaf node, whose \dosformula{indexing} contains a desired prefix (Fig.~\ref{simple_mpt_3}).
However, such a task is challenging.
This is because, for a leaf node storing blockchain state (e.g., account balance), its \dosformula{indexing} is derived from the keccak256 hash value~\cite{keccak256}  
of the information used to identify the leaf node (e.g., account address).
Hence,
to construct a leaf node whose \dosformula{indexing} contains a specific prefix, \dosattack{} needs to collide the 
\dosformula{indexing} of the leaf node, which is derived by keccak256 hash calculations~\cite{keccak256}.

Furthermore, in consideration of the trade-off between attack impact and cost,
we propose an optimized strategy to reduce \dosattack{}'s cost while retaining most of its original attack's impact (\S\ref{sec_design}). We achieve this by selectively deepening the leaf nodes that correspond to the most frequently accessed accounts (i.e., active accounts in \S\ref{sec_design}).

\ignore{There is a dilemma in colliding \dosformula{indexing} with a specific prefix.
Specifically,
constructing a deeper leaf node with a longer desired prefix leads to further more intermediate nodes, and causes extra resource consumption, but, it also requires more computing resources (e.g., GPUs) for adversaries to collide the prefix. 
Conversely, constructing a shallower leaf node with a shorter desired prefix causes limited resource consumption, although, it requires fewer computing resources. 
Besides, considering that there are billions of leaf nodes in \mpt{} (e.g., Ethereum's \mpt{}) for reserving blockchain state, it is non-trivial for \dosattack{} to proliferate intermediate nodes in all paths of \mpt{}.
To address the dilemma,
we choose to implement \dosattack{} to exploit the part of active accounts~\cite{ethanos} in \mpt{} (\S\ref{sec_design}).
This is because the leaf nodes, reserving the data of active accounts,
are located at the most frequently modified and accessed path of \mpt{}~\cite{ethanos}.
Hence, under this strategy, \dosattack{} can maximize the consumed resources of blockchain in maintaining and modifying the corresponding path and nodes in \mpt{}, under few resource cost of launching \dosattack{}.
}

\noindent\textbf{Evaluation.} 
To uniformly measure consumed resources of blockchain, e.g., CPU and disk, we utilize the time required for state modification as the metric~\cite{zheng18per},
and comprehensively and systematically evaluate \dosattack{} in four aspects.

First, we determine reasonable strategies of \dosattack{} with considering computing resources of launching attacks.
Please note that, with commodity hardware resources,
\dosattack{} needs to collide a keccak256 hash value with the specific prefix for manipulating \mpt{}. 
As a result,
under the computing resources of single RTX3080 GPU (Table~\ref{tab_preiamges}),
\dosattack{} can threaten \mpt{} by manipulating its structure at the depth of the first 15 layers (\S \ref{sec_evaluate_cap}), via colliding the first 13 nibbles of the hash values.

Second, 
we evaluate \dosattack{}'s impact on Ethereum from the block height of \#14.99M to \#15M.
As a result, with the computing resources of single RTX 3080 GPU, an adversary can persistently increase the time cost of state modification by 111\% of Ethereum (\S \ref{sec_evaluate_mainnet}) in a period of 10,000 blocks. %
We further propose models to estimate attack impact of \dosattack{}.
Our assessment shows that the models can estimate attack impact of \dosattack{} before launching it.

Third, we further evaluate the financial cost of \dosattack{} in exploiting seven popular blockchains. Table~\ref{tab_cost} enumerates the cost of \dosattack{}. It shows that the lowest cost of \dosattack{} is 39.64 USD while degrading the performance of Optimism in a period of 10,000 blocks.
Besides,
our optimization further reduces the cost of \dosattack{}.
Specifically,
by striking the active accounts in \mpt{} (\S\ref{sec_design}),
the cost of \dosattack{} can be further reduced
to 3.5\% of original cost with retaining 54.66\% of the original attack impact.
Our results indicate that, by targeting active accounts in \mpt{}, the adversary can optimize the cost of \dosattack{} to a reasonable range (\S \ref{sec_evaluate_eco}).

Fourth, we practically evaluate the effectiveness of \dosattack{} on Ethereum and BSC testnets. Please note that, unlike previous studies~\cite{ndss_broken_metro,heopartitioning,li2021deter} which only have non-persistent attack impacts, \dosattack{}'s impact persists in blockchain. 
Hence, due to ethical concerns, we evaluate the effectiveness of \dosattack{} in Ethereum and BSC testnets to minimize the potential negative impact. 
Besides, 
we carefully adjust attack parameters to light the attack impact. 
As a result,
when we witnessed that \dosattack{} caused the time of state modification on Ethereum (resp. BSC) testnet to increase by 15\% (resp. 18\%), we ceased the attack (\S \ref{sec_evaluate_testnet}).

\ignore{
\noindent\textbf{Implication.} We depict the four most serious consequences caused by \dosattack{} in the following. 
i) \dosattack{} persistently impairs the blockchain's performance for future blocks of the exploited blockchain. 
This is because,
even though \dosattack{}'s attack is finalized, the future state modification of the blockchain still suffers performance degradation once it involves intermediate nodes in \mpt{} proliferated by \dosattack{}.
ii) \dosattack{} delays blockchain users in using blockchain and AUX (e.g., flashbot~\cite{weintraub2022flash}, infura~\cite{li2021strong}, ENS~\cite{XiaENSIMC22}) in providing services, because \dosattack{} causes the waste of blockchain resource.
For example, 
blockchain users need to spend more 
resources in communicating with blockchain. 
Besides, 
an AUX like infura can only provide their services after they finish the delay of updating the latest state in \mpt{}.
iii) \dosattack{} threatens the consensus security of blockchain,
because it increases the cost of running blockchain nodes, thereby reducing the number of nodes served in the blockchain~\cite{heopartitioning}. 
iv) \dosattack{} also leads to a loss of trust in the blockchain platform, at least in the short term, resulting in a drop price of cryptocurrencies~\cite{ndss_broken_metro}.
}
At the time of writing, vulnerabilities under \dosattack{} have been confirmed by six blockchains (i.e., Ethereum, BSC, Polygon, Optimism, Avalanche, and Ethereum Classic), and we received thousands of USD bounty from them.

\noindent\textbf{Contributions} of this work are listed as follows

\begin{itemize}[leftmargin=*,topsep=1pt]

\item \textit{Novel attack at new attack surface.} Based on a new attack surface (i.e., the state storage of blockchain), we propose a novel DoS attack, \dosattack{}. By manipulating the MPT structure of state storage, \dosattack{} can persistently aggravate the consumed resources during blockchain state modification, including CPU, memory, and disk resources.

\item \textit{New observations.} To our best knowledge, we are the first to categorize the four classes of heavy time-consuming operations of blockchain.
Besides,
we reveal the flaw of gas mechanism, i.e., it fails to accurately reflect the actual consumed resources of state modification in \mpt{}.
Our two observations further inspire the design and the mitigation strategies of \dosattack{}.

\item \textit{New understandings.} We conduct a comprehensive and systematic evaluation of \dosattack{},
including 
assessing factors affecting \dosattack{},
evaluating the attack impact of \dosattack{},
measuring financial cost of \dosattack{} in exploiting seven mainstream blockchains, 
and validating the effectiveness of \dosattack{} on two testnets.
Our experimental results demonstrate that \dosattack{} can widely 
exploit various mainstream blockchain platforms, leading to a significant degradation of blockchain's performance at a reasonable cost.
Furthermore, we discuss the severe implications brought by \dosattack{}. 

\item \textit{Mitigations.} We elaborate on three classes of feasible mitigations that can reduce the attack impact of \dosattack{}, and discuss their advantages and disadvantages.

\end{itemize}

We release materials of our work in \url{https://github.com/hzysvilla/Nurgle_Oakland24} for future research.  

\section{Background}
\label{sec_background}

We introduce basic concepts of blockchain (\S\ref{sec_background_blockchain}), explain how blockchain adopts the MPT structure to maintain its state storage (\S\ref{sec_background_state}), and provide an example to illustrate how to exacerbate consumed resource in \mpt{} (\S\ref{sec_background_example}).

\subsection{Blockchain basic concepts}
\label{sec_background_blockchain}

We use the implementation of Ethereum~\cite{wood2014ethereum} to introduce the basic knowledge of blockchain, because \chain{} is a widely used blockchain platform. 
The native cryptocurrency of Ethereum is Ether.
According to its specification~\cite{wood2014ethereum}, the basic structure of its data is the block, which comprises the block header and block body. The block header involves a reference to the preceding block and the information used for state validation~\cite{wood2014ethereum}, while the block body contains a sequence of transactions. The transactions are signed by users to transfer funds and communicate with smart contracts.

There are two types of accounts in Ethereum, i.e., smart contract accounts (CA) and externally owned accounts (EOA). An EOA is controlled by a private key held by a user, while a CA holds the pre-defined logic and persistent variables. A contract is a Turing-complete automation program on Ethereum. The execution of contracts is facilitated by the Ethereum Virtual Machine (EVM), which is an underlying component of Ethereum supporting a set of instructions~\cite{wood2014ethereum}. 

The gas mechanism~\cite{gas_intro} establishes the cost associated with users utilizing the blockchain's resources. 
For example, each operation in Ethereum, which modifies the state data, will introduce the cost of gas,
e.g.,
executing contracts and transferring funds (e.g., Ether)~\cite{wood2014ethereum}.

\ignore{
\subsection{Blockchain}
\label{sec_background_blockchain}

We utilize the implementation of Ethereum~\cite{wood2014ethereum} to introduce the basic knowledge of blockchain. \chain{} is a widely used blockchain platform. According to its specifications~\cite{wood2014ethereum}, the basic structure of its data is the block, which comprises the header and body of the block. The block header involves a reference to the preceding block and the information used for state validation~\cite{wood2014ethereum}, while the block body contains a sequence of transactions. These transactions are signed by blockchain users to transfer funds and communicate with smart contracts.

There are two types of accounts in Ethereum, i.e., contract accounts (CA) and externally owned accounts (EOA). An EOA is controlled by a private key held by a user, while a CA holds the pre-defined logic and persistent variables. A contract is a Turing-complete automation program on Ethereum. The execution of contracts is facilitated by the Ethereum Virtual Machine (EVM), which is an underlying component of Ethereum supporting a set of instructions~\cite{wood2014ethereum}. 
Ethereum's native cryptocurrency is Ether. 

The gas mechanism~\cite{gas_intro} establishes the cost associated with users utilizing the blockchain's resources. 
For example, each operation in Ethereum, which modifies the state, will introduce the cost of gas,
e.g.,
executing contracts and transferring funds (e.g., Ether) will consume gas~\cite{wood2014ethereum}. 
}

\subsection{Blockchain state storage of MPT structure}
\label{sec_background_state}

In this study,
we focus on blockchains like Ethereum, whose state storage is organized in MPT structure (i.e., \mpt{}), comprising account data and contract data~\cite{ethanos}.

In Fig.~\ref{eth_stat}, we display the state storage of MPT structure (left), and its flat layout (right). Account Data \colorboxred{1} maintains all accounts’ information, e.g., balance, nonce, code hash, and storage root of each account. 
Besides, the account information of each account is separately reserved in a unique leaf node.
In Account Data \colorboxred{1}, the account information of each account is mapped by the \dosformula{indexing} of the leaf node reserving the account information.
As mentioned in \S\ref{sec_intro}, the \dosformula{indexing} (e.g., \doscode{abc6d}) is used to locate a leaf node in the MPT structure~\cite{mptanalysis}.
If an account is a contract, its storage root \colorboxred{5} in Account Data \colorboxred{1} points to its Contract Data \colorboxred{2}. 
Contract Data \colorboxred{2} contains the data of the contract to store persistently~\cite{wood2014ethereum}.
Specifically,
Contract Data \protect\colorboxred{2} is a 
mapping from the slots~\cite{ayub2023storage} of a contract's storage to the values stored in corresponding slots~\cite{popular_con}.
In MPT structure, each data, mapping to a contract storage slot, is reserved in a unique leaf node, and the \dosformula{indexing} of the leaf node is derived from its slot.
In addition,
each contract has an independent
Contract Data \colorboxred{2}~\cite{wood2014ethereum}.
Account Data \colorboxred{1} and Contract Data \colorboxred{2} are maintained in State Trie \colorboxred{3} and Storage Tries \colorboxred{4}, respectively.
Besides, both State Trie \colorboxred{3} and Storage Tries \colorboxred{4} are in MPT structure for managing and verifying state.

MPT structure manages and indexes data by compressing the prefix tree~\cite{szpankowski1990patricia}.
In \mpt{}, the \dosformula{indexing} of leaf nodes is a 256-bit byte array~\cite{wood2014ethereum}, i.e., containing 64 nibbles.
Please note that, we denote the length of \dosformula{indexing} as the number of nibbles (i.e., half byte) in it.
Taking the \dosformula{indexing} \doscode{abc6d} (Fig.~\ref{eth_stat}) as an instance, the length of its \dosformula{indexing} \doscode{abc6d} is 5.
To obtain the data reserved in the leaf node, we %
need to
locate the leaf node with the \dosformula{indexing} \doscode{abc6d} by searching for the intermediate nodes holding the prefix of \doscode{abc6d} (i.e., the three intermediate nodes \doscode{a}, \doscode{bc}, and \doscode{6d} in Fig.~\ref{eth_stat}).

\begin{figure}[!t]

\small
	\centering
	\includegraphics[width=0.99\linewidth]{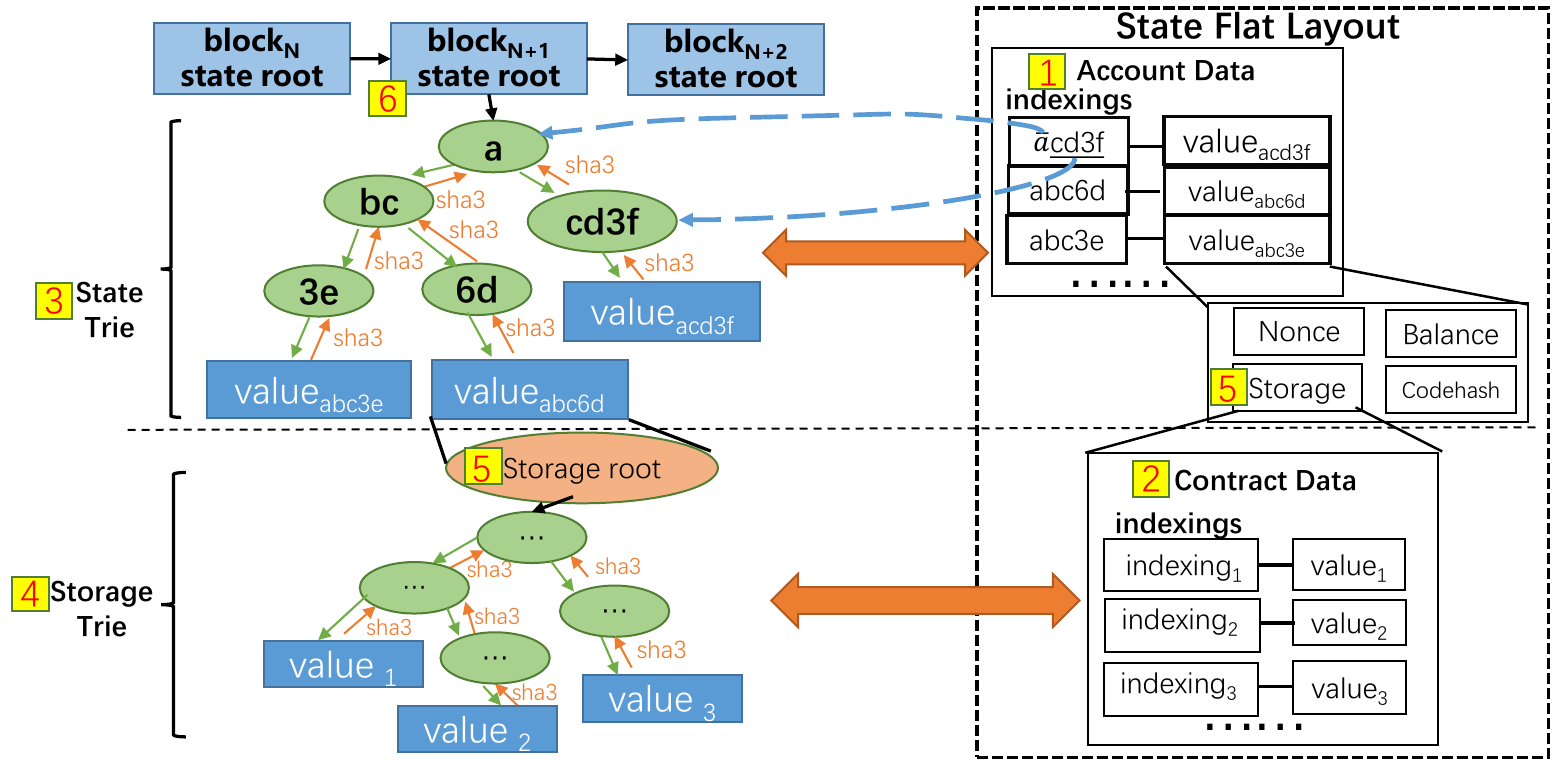}
	\caption{The left part shows state storage in the MPT structure. We display the flat layout of corresponding information at the right part to facilitate intuitive understanding.
}
	\label{eth_stat}
\end{figure}

\mpt{} performs the state verification for consensus in the idea of the Merkle tree~\cite{szydlo2004merkle}. Each parent node maintains the keccak256 hash value of all its child nodes' data. The keccak256 hash value of the root node atop State Trie \colorboxred{3} is used as the metadata of a block's header, named as state root \colorboxred{6}. The orange arrow in Fig.~\ref{eth_stat} displays the procedure of verifying state. When the state of an account (e.g., Ether balance) changes, since the account's data is stored on a leaf node, it will force all the hash values from the leaf node to the root node to be changed, and the information within nodes in the path is required to be updated accordingly. Taking the account with the \dosformula{indexing} \doscode{acd3f} in Fig.~\ref{eth_stat} as an example, when the information of the account changes, the hash values maintained by nodes \doscode{a} and \doscode{cd3f} will be changed, resulting in both the two nodes to be updated.

\begin{figure}[!t]
\small
	\centering
	\includegraphics[width=0.74\linewidth]{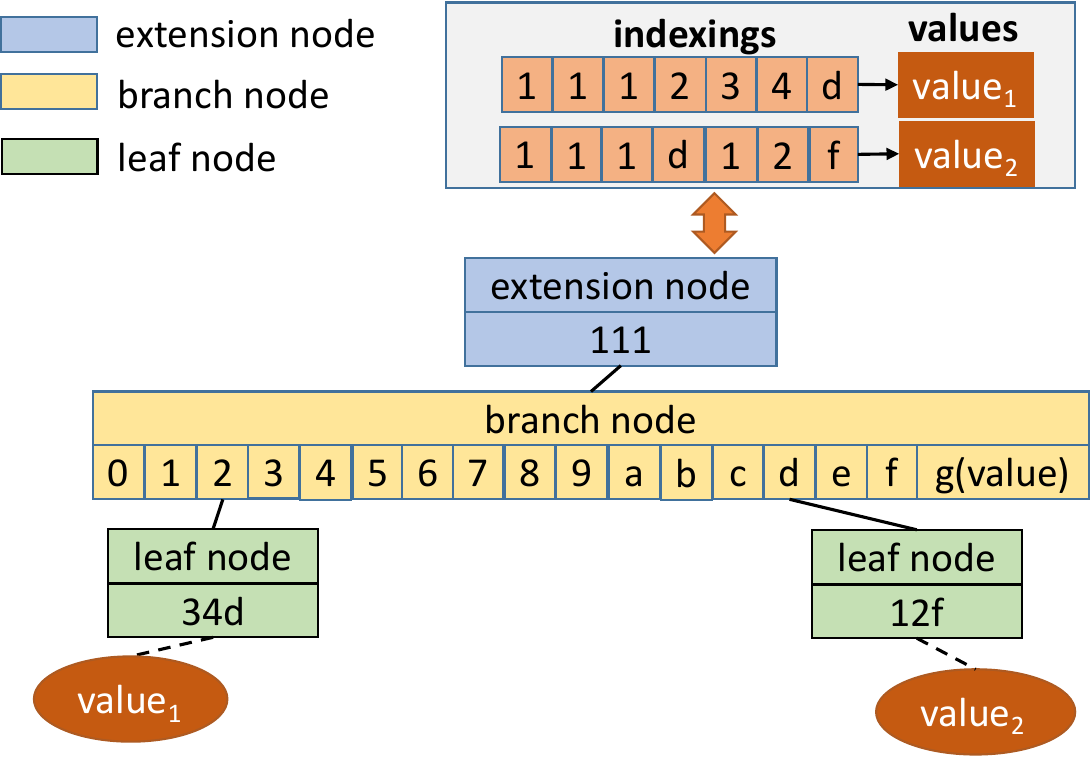}
        \vspace{2pt}
	\caption{A simplified \mpt{}, constructed by two leaf nodes with \dosformula{indexing} as \doscode{111234d} and \doscode{111d12f}. The extension node (\doscode{111}) handles the common prefix (\doscode{111}) of the two leaf nodes' \dosformula{indexing}. The branch node forks to point the two leaf nodes by two pointers (i.e., \doscode{2} and \doscode{d}). Finally, the two leaf nodes keep the unique part of their \dosformula{indexing} as \doscode{34d} and \doscode{12f}.
}
	\label{simple_mpt_}
\end{figure}

\subsection{Exacerbating consumed resource in \mpt{}}
\label{sec_background_example}

We depict the detailed design of \mpt{} in \chain{}, and illustrate an example of MPT manipulation, which causes the resource consumption of \chain{} to be exacerbated.
There are 
three types of nodes in \chain{}'s \mpt{}, i.e., branch nodes, extension nodes, and leaf nodes. A branch node stores up to 16 pointers (from pointer \doscode{0} to pointer \doscode{f}).
A pointer can point to a leaf node, an extension node, or another branch node. An extension node reserves a byte sequence, which is used to point to nodes with the common hash sequence. An extension node also contains a pointer to its child node. A leaf node stores a pointer targeting the data reserved in the leaf node (e.g., Ether balance). We deliver a simplified \mpt{} in Fig.~\ref{simple_mpt_}, where there are two leaf nodes storing data. The \dosformula{indexing} of the two leaf nodes are \doscode{111234d} and \doscode{111d12f}. The common prefix of the two \dosformula{indexing} is \doscode{111}, which is handled by the extension node \doscode{111}. The branch node forks to point the two leaf nodes by utilizing the two pointers \doscode{2} and \doscode{d}, respectively.
Finally, the two leaf nodes keep the unique part of their \dosformula{indexing}, i.e., \doscode{34d} and \doscode{12f}.

Inserting and deleting nodes can trigger the node splitting and collapse in \mpt{}, which leads to extra resource consumption~\cite{xiaoju2020ebtree}. Fig.~\ref{simple_mpt_3} illustrates the node splitting triggered by inserting a leaf node (whose \dosformula{indexing} is \doscode{111d1f3}) into the \mpt{}. Since the \dosformula{indexing} \doscode{111d1f3} of the inserted leaf node has a common prefix (i.e., \doscode{111d1}) with the \dosformula{indexing} \doscode{111d12f} of an existing leaf node in the \mpt{}. Therefore, the leaf node (i.e., \doscode{12f}) is split into an extension node \doscode{1}, a branch node (containing two pointers as \doscode{2} and \doscode{f}), and two leaf nodes (i.e., \doscode{f} and \doscode{3}).
During the whole procedure of inserting the leaf node \doscode{111d1f3}, it additionally consumes resources (e.g., CPU, memory, and disk resources) for maintaining and verifying the newly generated intermediate nodes.
Furthermore,
the node collapse will happen by deleting leaf nodes from \mpt{}.
Taking the right part of Fig.~\ref{simple_mpt_3} as an example, during deleting the leaf node \doscode{111d1f3} from the \mpt{}, it additionally consumes resources to discard nodes (which are marked with red dotted lines), and re-verify all involved nodes.

\section{Threat model}
\label{sec_model}

\begin{figure}[!t]
\small
	\centering
	\includegraphics[width=0.99\linewidth]{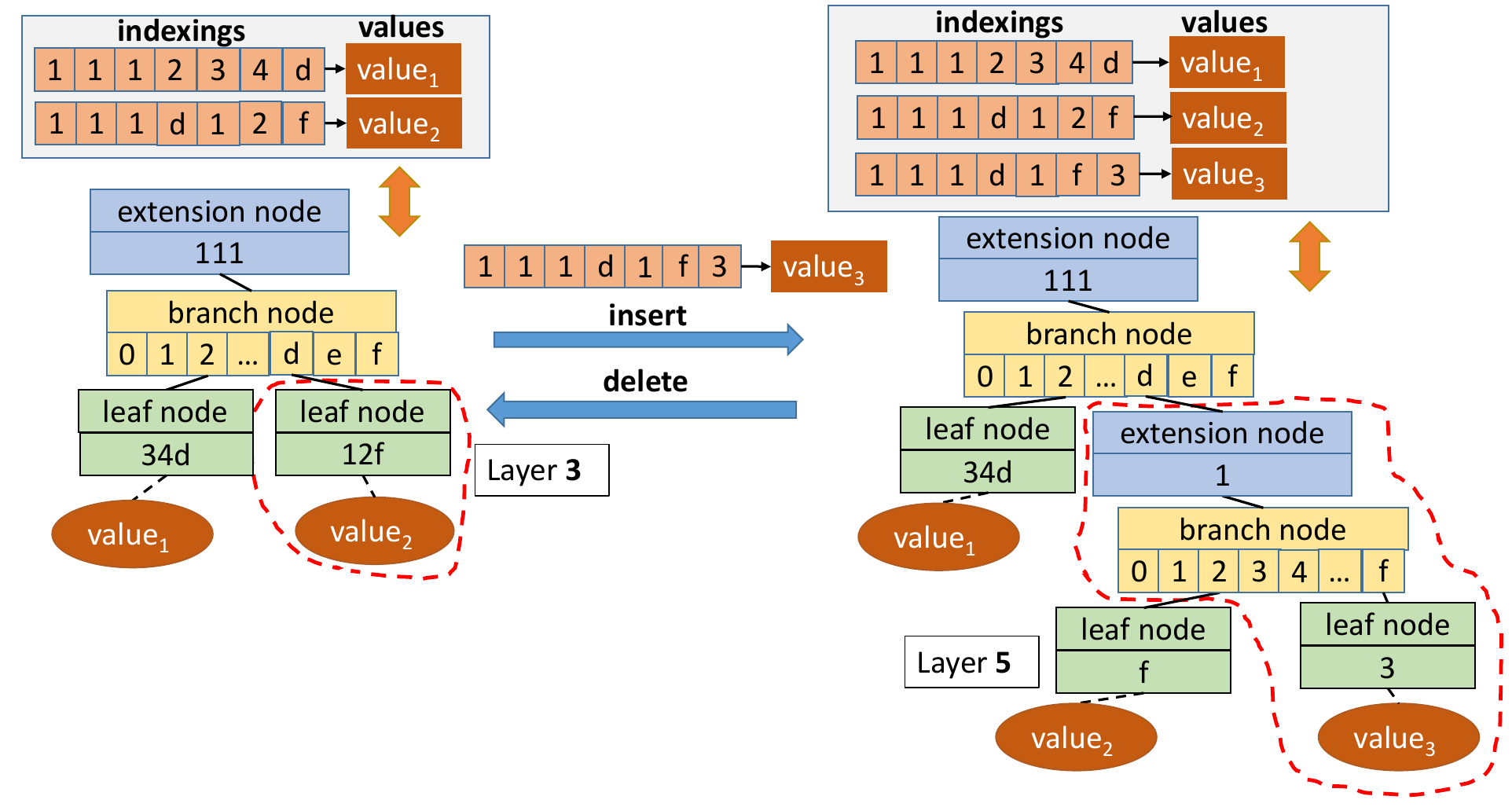}
	\caption{After inserting a leaf node \doscode{111d1f3} into the \mpt{}, the leaf node \doscode{12f} splits into an extension node, a branch node, and two leaf nodes. Finally, the \mpt{} is proliferated with new nodes, and its depth is deepened into two layers. Newly generated nodes are marked with red dotted lines.}
	\label{simple_mpt_3}
\end{figure}

The threat model of \dosattack{} involves two actors: an adversary and a victim blockchain.
The adversary submits crafted payload by transactions to the victim blockchain, resulting in impairing the blockchain’s performance.
In the threat model,
the victim blockchain supports the execution of smart contracts,
and adopts the MPT structure to maintain and update its state storage.
Besides, there is a P2P network atop the victim blockchain, accepting users to submit transactions to be included in the victim blockchain.

The adversary controls an externally owned account (EOA) and necessary assets for sending transactions to the P2P network of the victim blockchain.
Besides, the adversary controls a modified client of the victim blockchain, by which the adversary can monitor and analyze the status quo of the victim blockchain's \mpt{} to construct the attack payload.
Moreover, the adversary has limited resources, e.g., GPUs, for constructing the payload to mount attack.

It is reasonable for a financially rational adversary to launch \dosattack{} for two reasons. i) \dosattack{} can yield economic returns by arbitraging during the price volatility of cryptocurrency assets on the victim blockchain~\cite{ndss_broken_metro}. For example, the adversary can first utilize \dosattack{} to impair the performance of the victim blockchain and weaken the trust of related ecology, leading to a decline in the asset price of the cryptocurrency. After that, the adversary can just conduct the short selling~\cite{saffi2011price} on the corresponding cryptocurrency assets to obtain a considerable profit. ii) Adversaries running the blockchain can utilize \dosattack{} to disrupt competitors. By striking competitors, it drives the customer base of the victim, flocking to the adversaries~\cite{ndss_broken_metro}.

The cost of \dosattack{} consists of two parts, i.e., the fee of computing resources and the gas fee. The major fee of computing resources comes from GPUs, which are used to collide a leaf node with a desired prefix for inserting it into a target position of the \mpt{} (cf. details in \S \ref{sec_design}). It is necessary to undertake gas fee for the adversary because the adversary needs to submit transactions with crafted payload in the victim blockchain for launching \dosattack{} (\S \ref{sec_design}).

\section{Observation of blockchain state storage}
\label{sec_observation}

In this section, we illustrate two observations, i.e., the heavy burden of state maintenance (\S \ref{sec_observation_stateop}) and the flaw of gas mechanism (\S \ref{sec_observation_gasflaw}), which inspires the design of \dosattack{}.

\begin{table*}[!t]
\caption{Four operations that involve heavy burden of state maintenance}
\label{tab_stateops}
\centering
\resizebox{0.99\linewidth}{!}{
\begin{tabular}{@{}l|l|l@{}}
\toprule
Operation type                           & Description  & Major consumed resources                                                                                                 \\ \midrule \midrule 
\textbf{[OP1]} MPT update              & 
\begin{tabular}[c]{@{}l@{}}During the state modification, the \mpt{} will update nodes in \mpt{} (i.e., nodes\\ in State Trie \colorboxred{3}, and Storage Tries \colorboxred{4}) that are required to be
modified. \end{tabular}         &      Disk read, and
memory read and write                \\ \midrule
\textbf{[OP2]} MPT verification           & 
\begin{tabular}[c]{@{}l@{}}To verify whether \mpt{} holds the latest state, \mpt{} derives the keccak256 hash\\ value of the root node of the whole State Trie \colorboxred{3} from all other nodes in \mpt{}\\ at a bottom-up manner. \end{tabular} & CPU computation  \\ \midrule
\textbf{[OP3]} MPT holding in memory  & 
\begin{tabular}[c]{@{}l@{}}
To mitigate consumed disk resource in accessing nodes in \mpt{}, the blockchain\\ holds partial nodes in \mpt{} (e.g., nodes in State Trie \colorboxred{3} and Storage Tries \colorboxred{4}) in\\ memory, and can determine which nodes holding in memory to be discarded.
\end{tabular}
& CPU computation, and memory read and write
\\ \midrule
\textbf{[OP4]} MPT persistence              &  \mpt{} persistently stores nodes representing the latest state of blockchain into disk. & Disk write \\ \bottomrule
\end{tabular}
}
\vspace{4pt}
\end{table*}

\subsection{Heavy burden of state maintenance}
\label{sec_observation_stateop}

Operations with \mpt{} (e.g., maintaining, verifying, modifying, and assessing state in \mpt{}) are the major 
performance bottleneck of blockchain.
Previous studies~\cite{Li2023LVMT,ponnapalli2021rainblock,ethanos} report that over 81\% execution time of blockchain costs in interacting with \mpt{}.
To make it worse, we further uncover that the time cost of major time-consuming operations in interacting with \mpt{} linearly increases with the number of nodes in \mpt{} involved in these operations, leading the performance of blockchain heavily depend on the number of nodes in \mpt{}.
We categorize the operations into four classes in Table~\ref{tab_stateops} and depict them as follows.

\begin{itemize}[leftmargin=*]
    \item \textbf{OP1 (MPT update).} Blockchain updates nodes in \mpt{} involved in state modification (e.g., node splitting in~\S\ref{sec_background_example}). 
    In \textbf{OP1}, both nodes, corresponding to Account Data \colorboxred{1} in State Trie \colorboxred{3} and Contract Data \colorboxred{2} in Storage Tries \colorboxred{4}, can be modified.
    For example, when an account's balance changes, the leaf node (corresponding to the account) in State Trie \colorboxred{3} will update its reserved data, e.g., updating to the account's latest balance.
    \textbf{OP1} consumes disk read, and memory write and read, because intermediate nodes and leaf nodes associated with the account can be absent in current memory, and they have to be accessed from the disk accordingly in such cases.
    \item \textbf{OP2 (MPT verification).} According to \S\ref{sec_background_state}, to achieve the consensus of blockchain state, \mpt{} needs to derive the keccak256 hash value of the root node of the whole State Trie \colorboxred{3} from all other nodes in \mpt{} at a bottom-up manner (\S\ref{sec_background_state}). We denote the operation for computing the hash value of the root node as MPT verification (\textbf{OP2}), since the hash value is used to verify whether \mpt{} holds the latest state during consensus.
    Specifically,
    \textbf{OP2} will first derive the hash value of the root of each Storage Trie \colorboxred{4} from all nodes in the Storage Trie \colorboxred{4}, and then calculate the hash value of the root node of the whole State Trie \colorboxred{3} (\S\ref{sec_background_state}). 
    \textbf{OP2} consumes expensive CPU resource to calculate the keccak256 hash value for each node~\cite{ponnapalli2021rainblock}. 
    \item \textbf{OP3 (MPT holding in memory).} 
    To ease consumed disk resource in accessing nodes in \mpt{}, blockchain holds partial nodes of \mpt{} (e.g., nodes in State Trie \colorboxred{3} and Storage Tries \colorboxred{4}) in memory, and can determine which nodes holding in memory to be discarded from memory. We denote the operation for holding nodes of \mpt{} in memory as \textbf{OP3}.
    Besides, 
    discarded nodes are expired nodes that have been replaced by other latest nodes during state changes~\cite{state_prune}. 
    Since discarded nodes will not be written to disk, \textbf{OP3} mainly consumes CPU and memory resources for maintaining nodes of \mpt{} in memory.
    \item \textbf{OP4 (MPT persistence).} The \mpt{} persistently stores nodes representing the latest state of blockchain into disk (\textbf{OP4}), and generates heavy overhead of disk writes. 
\end{itemize}

\begin{figure}[!b]
\small
	\centering
	\includegraphics[width=0.96\linewidth]{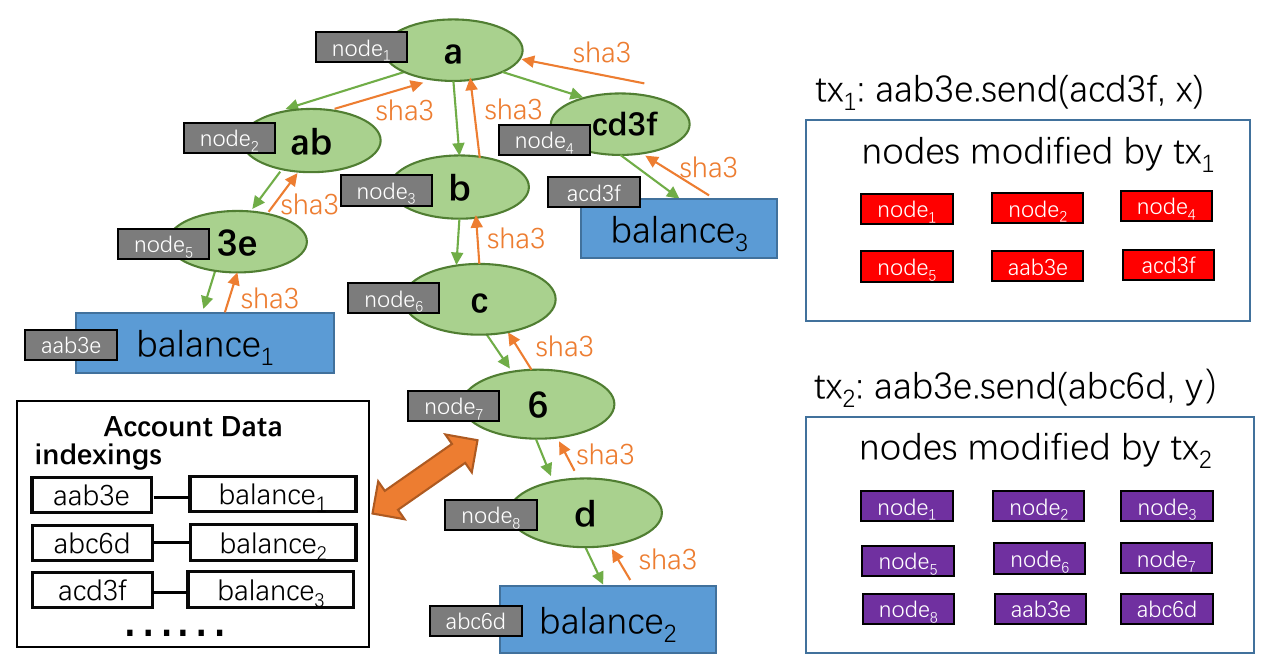}
	\caption{There are two transactions \dosformula{$\dosformula{tx}_1$} and \dosformula{$\dosformula{tx}_2$} transferring Ether, and each of them consumes 21,000 units of gas. However, when executing \dosformula{$\dosformula{tx}_1$} and \dosformula{$\dosformula{tx}_2$}, the number of nodes in \mpt{} to be updated (\S\ref{sec_background_state}) is different (marked as red and purple boxes in Fig.~\ref{gas_incon_demo}). After finalizing the two transactions, although the two transactions cost the same amount of gas, they consume different resources for updating the \mpt{}.}
	\label{gas_incon_demo}
\end{figure}

Please note that, \textbf{OP1-4} can trigger each other.
For Fig.~\ref{simple_mpt_3} as an example, 
once the intermediate nodes and leaf nodes are modified in \textbf{OP1}, the keccak256 hash value of them should be updated, including the root node of the whole State Trie \colorboxred{3} (\textbf{OP2}). Besides, since new nodes are generated in \mpt{}, the nodes of \mpt{} in memory can also be replaced and discarded (\textbf{OP3}). Furthermore, the newly generated nodes can also be written into disk for persistently maintaining (\textbf{OP4}).
In addition, 
the time complexity to finalize all \textbf{OP1-4} is $\mathcal{O}(n)$, where $n$ is the number of nodes in \mpt{} involved in \textbf{OP1-4}.
It means that the more MPT nodes need to be maintained, the more resources will be consumed.
Our first observation assists \dosattack{} to exacerbate consumed resources by involving more nodes in \mpt{} to increase more time cost for handling the four classes of operations.

\subsection{Flaw of the gas mechanism}
\label{sec_observation_gasflaw}

To our best knowledge, we are the first to uncover the flaws of gas mechanism in the view of state storage. 
We use an example in Fig.~\ref{gas_incon_demo} to illustrate the flaw of gas mechanism revealed by us. There are three EOAs in Fig.~\ref{gas_incon_demo}. The leaf nodes reserving the three EOAs' information are located by their \dosformula{indexing} as \doscode{aab3e}, \doscode{abc6d}, and \doscode{acd3f} (marked as blue boxes). 
In the example, there are two transactions (i.e., \dosformula{$\dosformula{tx}_1$} and \dosformula{$\dosformula{tx}_2$}), where \dosformula{$\dosformula{tx}_1$} transfers Ether from \doscode{aab3e} to \doscode{acd3f}, and \dosformula{$\dosformula{tx}_2$} transfers Ether from \doscode{aab3e} to \doscode{abc6d}.
According to Ethereum's specifications~\cite{wood2014ethereum}
, each of \dosformula{$\dosformula{tx}_1$} and \dosformula{$\dosformula{tx}_2$} costs 21,000 units of gas. 
However, the resources consumed by the two transactions are different. 
Specifically, 
during the execution of \dosformula{$\dosformula{tx}_1$}, six nodes (marked as red boxes) in the \mpt{} are required to be updated.
However, during the execution of \dosformula{$\dosformula{tx}_2$}, nine nodes (marked as purple boxes) are required to be updated.
Therefore, although the two transactions cost the same amount of gas, the amount of resources consumed by them is different for updating the \mpt{}.
Similarly, the flaw of gas mechanism can also be observed in updating the data reserved in contract storage (Storage Tries \colorboxred{4}). The root cause for the flaw is that the current gas mechanism (e.g., gas mechanism of Ethereum~\cite{wood2014ethereum}) does not consider the exact resources consumed for maintaining and verifying state in \mpt{} (e.g., the resources for modifying intermediate nodes). 
Our second observation inspires \dosattack{} to exacerbate consumed resources of state modification by introducing more intermediate nodes involved in the state modification.

\section{The design and implementation of \dosattack{}}
\label{sec_design_implement}

\begin{figure*}[!t]
\small
  \begin{subfigure}{0.33\linewidth}
    \centering
    \includegraphics[width=0.99\linewidth]{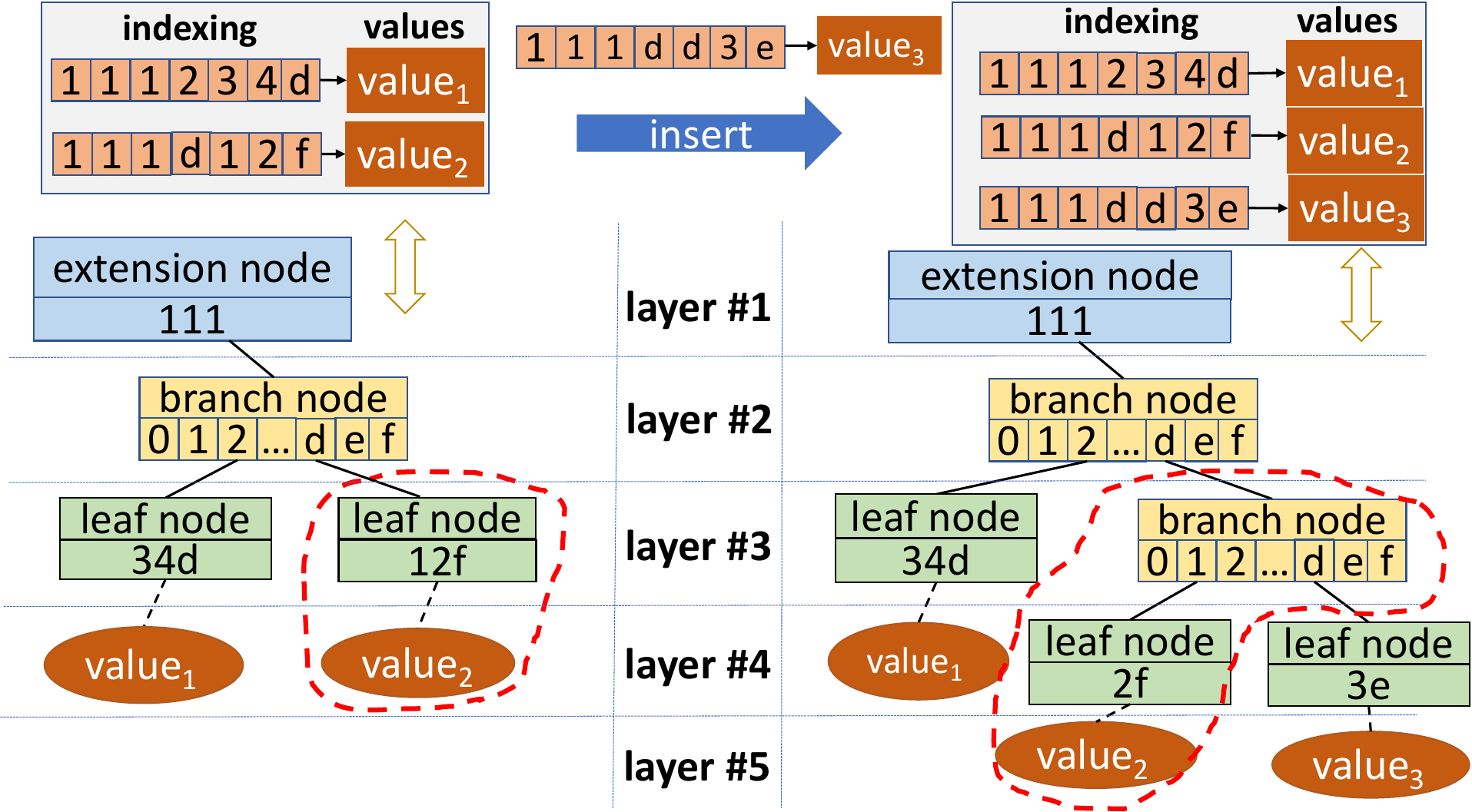}
    \caption{An example of Strategy 1}
    \label{fig_split_case1}
  \end{subfigure}
  \begin{subfigure}{0.33\linewidth}
    \centering
    \includegraphics[width=0.99\linewidth]{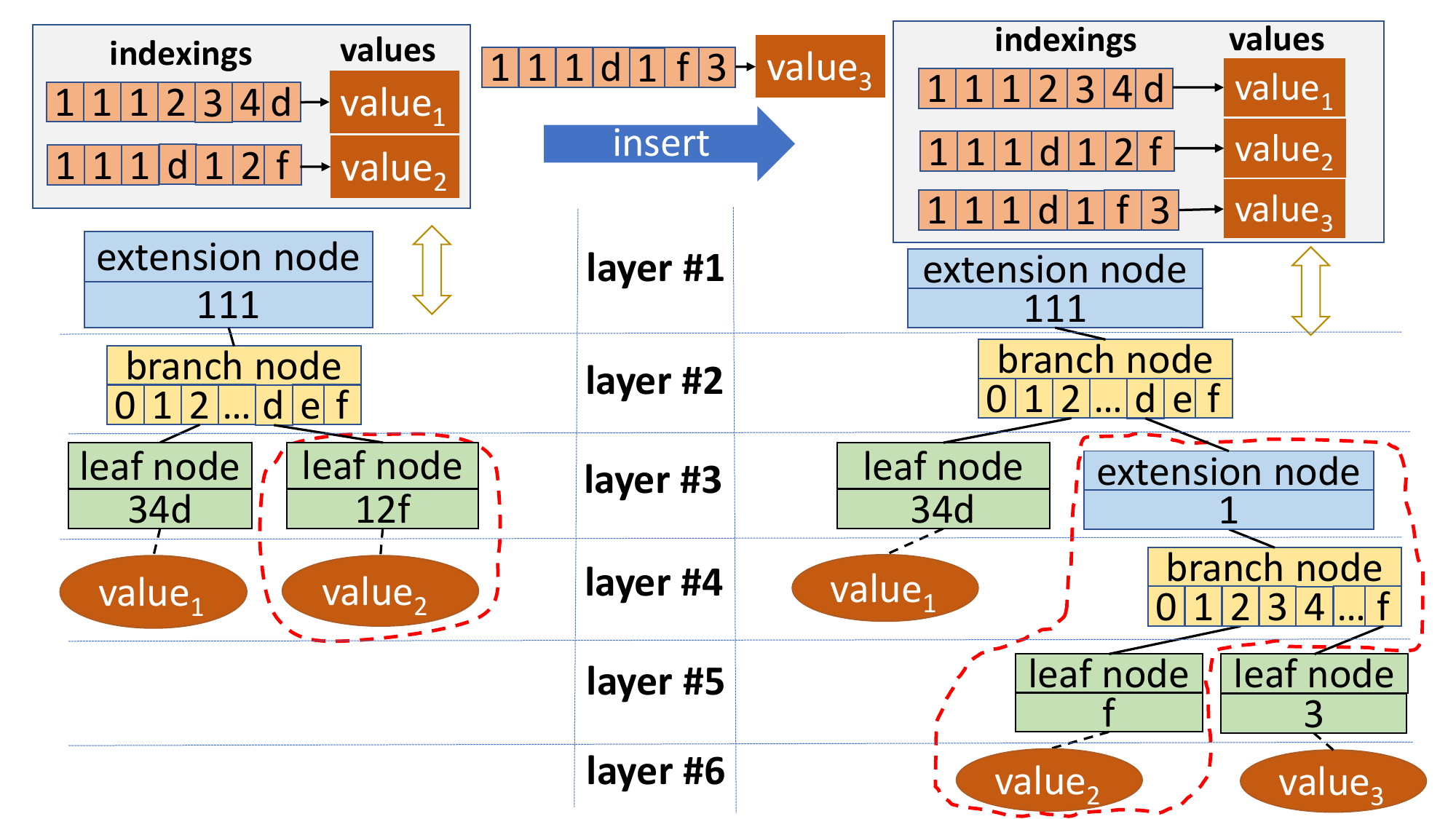}
    \caption{An example of Strategy 2}
    \label{fig_split_case2}
  \end{subfigure}
   \begin{subfigure}{0.33\linewidth}
    \centering
    \includegraphics[width=0.99\linewidth]{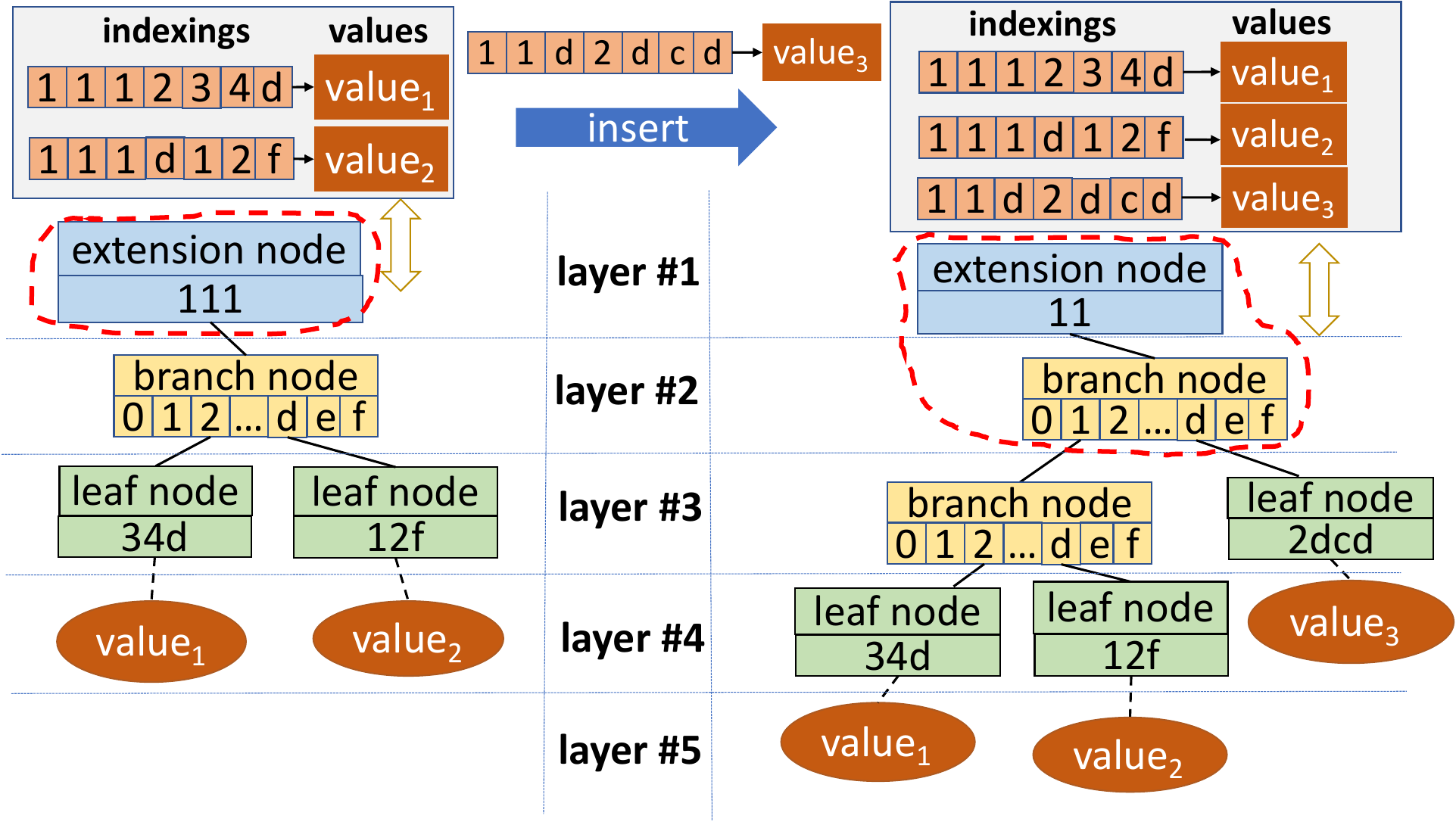}
    \caption{An example of Strategy 3}
    \label{fig_split_case3}
  \end{subfigure}
  
  \vspace{2pt}
  \caption{Examples of utilizing the three strategies (\textbf{S1-3}). Nodes before and after splitting are marked with red dotted lines.}
  \vspace{1pt}
  \label{fig_split_cases}
\end{figure*}

\begin{algorithm}[!b]
\caption{\dosattack{}}
\label{alg_nurgle}
\SetKwRepeat{Do}{do}{while}
\scriptsize{
  \KwIn{\dossmallformula{node$_{\dossmallformula{v}}$}, the leaf node to be deepened in \mpt{}}
  \KwIn{\dossmallformula{node$_{\dossmallformula{root}}$}, the root node of \mpt{}}
  \KwOut{\dossmallformula{nodes$_{\dossmallformula{insert}}$}, the leaf nodes to be inserted for deepening \dossmallformula{node$_{\dossmallformula{v}}$}}

\SetArgSty{text} 
\dossmallformula{nodes$_{\dossmallformula{insert}}$} $\leftarrow$ []\\
\dossmallformula{nodes$_{\dossmallformula{collided}}$} = \texttt{\scriptsize HashCollision(}\dossmallformula{node$_{\dossmallformula{v}}$}\texttt{\scriptsize )}

\Do{New \dossmallformula{node$_{\dossmallformula{insert}}$} has appended into \dossmallformula{nodes$_{\dossmallformula{insert}}$}}{
    \dossmallformula{nodes$_{\dossmallformula{intermediate}}$} = \texttt{\scriptsize TraverseNodes(}\dossmallformula{node$_{\dossmallformula{root}}$}, \dossmallformula{node$_{\dossmallformula{v}}$}\texttt{\scriptsize )} \\
    \For{\dossmallformula{node} $\in$ \dossmallformula{nodes$_{\dossmallformula{intermediate}}$}}{  
    \uIf{\texttt{\scriptsize Type(}\dossmallformula{node}\texttt{\scriptsize )} \textbf{is} "extension"} {
         \uIf{\texttt{\scriptsize IsSplittable(}\dossmallformula{node}\texttt{\scriptsize )}} {
            \dossmallformula{node$_{\dossmallformula{insert}}$} = \texttt{\scriptsize MatchS3(}\dossmallformula{node}, \dossmallformula{nodes$_{\dossmallformula{collided}}$}\texttt{\scriptsize )}\\
            \uIf{\dossmallformula{node$_{\dossmallformula{insert}}$} $\neq$ \texttt{\scriptsize null}}{
                \dossmallformula{nodes$_{\dossmallformula{insert}}$}.append(\dossmallformula{node$_{\dossmallformula{insert}}$})
            }
        }
    }
    \uElseIf{\texttt{\scriptsize Type(}\dossmallformula{node}\texttt{\scriptsize )} \textbf{is} "leaf"}{
        \uIf{\texttt{\scriptsize IsSplittable(}\dossmallformula{node}\texttt{\scriptsize )}} {
            \dossmallformula{node$_{\dossmallformula{insert}}$} = \texttt{\scriptsize S2Match(}\dossmallformula{node}, \dossmallformula{nodes$_{\dossmallformula{collided}}$}\texttt{\scriptsize )}\\
            \uIf{\dossmallformula{node$_{\dossmallformula{insert}}$} $\neq$ \texttt{\scriptsize null}}{
                \dossmallformula{nodes$_{\dossmallformula{insert}}$}.append(\dossmallformula{node$_{\dossmallformula{insert}}$})
            }
            \Else{
                \dossmallformula{node$_{\dossmallformula{insert}}$} = \texttt{\scriptsize S1Match(}\dossmallformula{node}, \dossmallformula{nodes$_{\dossmallformula{collided}}$}\texttt{\scriptsize )}\\
            \uIf{\dossmallformula{node$_{\dossmallformula{insert}}$} $\neq$ \texttt{\scriptsize null}}{
                \dossmallformula{nodes$_{\dossmallformula{insert}}$}.append(\dossmallformula{node$_{\dossmallformula{insert}}$})
            }
            }
        }
    }
    
    }
}
\KwRet \dossmallformula{nodes$_{\dossmallformula{insert}}$}
}
\end{algorithm}

In this section, we elaborate on the design of \dosattack{} (\S \ref{sec_design}) and \dosattack{}'s detailed implementation (\S \ref{sec_implement}).

\subsection{The design of \dosattack{}}
\label{sec_design}

Inspired by our two observations (\S\ref{sec_observation}),
\dosattack{} aims to expand intermediate nodes of \mpt{} by inserting leaf nodes into desired positions of \mpt{}. 
Specifically, to expand the intermediate nodes in \mpt{}, we utilize the node splitting (\S\ref{sec_background_example}) triggered while inserting leaf nodes into \mpt{}. 
After that, the expanded intermediate nodes will increase the consumed resources for the operations of modifying, updating, and verifying nodes in \mpt{} (e.g., \textbf{OP1-4} in \S\ref{sec_observation_stateop}).

Inserting a leaf node in \mpt{} by different strategies can trigger node splitting in different ways, causing distinct results (\S\ref{sec_background_example}). In the following, we categorize three strategies, which split nodes and deepen a target leaf node (\dosformula{node$_{v}$}) by inserting a leaf node (\dosformula{node$_{\dossmallformula{insert}}$}) in different ways.
We denote that, the length of the whole \dosformula{indexing} of \dosformula{node$_{\dossmallformula{v}}$} is \dosformula{m}, and the length of the unique part of \dosformula{node$_{v}$} is \dosformula{n}. For the leaf node \doscode{acd3f} in Fig.~\ref{eth_stat} as an example, the length of its whole \dosformula{indexing} and the length of its unique part are 5 and 4, respectively.

\noindent
$\bullet$ \textbf{S1}.
\dosformula{node$_{\dossmallformula{v}}$} splits into a leaf node, and a branch node, when the length of the common prefix between the \dosformula{indexing} of \dosformula{node$_{\dossmallformula{v}}$} and \dosformula{node$_{\dossmallformula{insert}}$} equals to \dosformula{m-n}.
\textbf{S1} deepens \dosformula{node$_{v}$} by 1 layer, and adds an intermediate node. 
For example,
when inserting the leaf node \doscode{111dd3e} (Fig.~\ref{fig_split_case1}), leaf node (whose unique part is \doscode{12f}) splits into a branch node (containing the two pointers \doscode{1} and \doscode{d}), and a leaf node (whose unique part is \doscode{2f}). 

\noindent
$\bullet$ \textbf{S2}.
\dosformula{node$_{\dossmallformula{v}}$} splits into an extension node, a branch node, and a new leaf node, when the length of the common prefix between the \dosformula{indexing} of \dosformula{node$_{\dossmallformula{\dossmallformula{v}}}$} and \dosformula{node$_{\dossmallformula{insert}}$} is larger than \dosformula{m-n}.
\textbf{S2} deepens \dosformula{node$_{v}$} by 2 layers, and adds two intermediate nodes.
For example,
when inserting the leaf node \doscode{111d1f3} (Fig.~\ref{fig_split_case2}), the leaf node (whose unique part is \doscode{12f}) splits into an extension node \doscode{1}, a branch node,
and a leaf node (whose unique part is \doscode{f}). 

\noindent
$\bullet$ \textbf{S3}.
An extension node in the path from the root node to \dosformula{node$_{v}$}
splits into a branch node and a new extension node, when i) the common prefix between the \dosformula{indexing} of \dosformula{node$_{\dossmallformula{\dossmallformula{v}}}$} and \dosformula{node$_{\dossmallformula{insert}}$} cannot cover the prefix maintained in the extension node, and ii) the length of the prefix maintained in the extension node is larger than 1.
\textbf{S3} deepens \dosformula{node$_{v}$} by 1 layer, and adds an intermediate node.
For example,
when inserting leaf node \doscode{11d2dcd} (Fig.~\ref{fig_split_case3}), extension node (maintaining \doscode{111}) splits into an extension node \doscode{11} and a branch node.

Algorithm~\ref{alg_nurgle} presents the process of \dosattack{}.
\dosattack{} crafts a list of leaf nodes (i.e., \dosformula{nodes$_{\dossmallformula{insert}}$}). 
By inserting in \mpt{},
\dosformula{nodes$_{\dossmallformula{insert}}$} expand intermediate nodes by node splitting and deepen \dosformula{node$_{\dossmallformula{v}}$}.
Specifically,
for a given leaf node \dosformula{node$_{\dossmallformula{v}}$}, under a predefined timeout,
\dosattack{} collides the \dosformula{indexing} of \dosformula{node$_{\dossmallformula{v}}$} with aiming to maximize the length of common prefix between \dosformula{node$_{\dossmallformula{v}}$} and collided nodes (Line 2).
After the timeout, \dosattack{} collects all collided nodes in \dosformula{node$_{\dossmallformula{collided}}$}.
In Line 4,
\dosattack{} then retrieves all nodes in the path of \mpt{} from the root node of \mpt{} to \dosformula{node$_{\dossmallformula{v}}$} (i.e., \dosformula{nodes$_{\dossmallformula{intermediate}}$}).
\dosattack{} iterates nodes in \dosformula{nodes$_{\dossmallformula{intermediate}}$} to determine whether they can be split by node splitting (Line 7 and 12).
If a node in \dosformula{nodes$_{\dossmallformula{intermediate}}$} can be split, \dosattack{} generates the leaf node \dosformula{node$_{\dossmallformula{insert}}$} for splitting the node, by matching whether the three strategies are satisfied (Line 8, 13, and 17).
\dosattack{} then collects \dosformula{node$_{\dossmallformula{insert}}$} into \dosformula{nodes$_{\dossmallformula{insert}}$} (Line 10, 15 and 19).
\dosattack{} will continue the whole procedure (Line 4 - Line 19), until there are no more new leaf nodes for node splitting.
It is worth noting that, during the whole procedure (Algorithm~\ref{alg_nurgle}), to improve the efficiency of \dosattack{}, we also record the nodes that can not be split, which helps \dosattack{} to directly skip these nodes in Line 7 and 12.

For a leaf node in \mpt{}, Lemma~\ref{lem_attack} guarantees that \dosattack{} can deepen the leaf node by triggering leaf splitting. 

\begin{lemma}
\label{lem_attack}
Given a leaf node \dosformula{node}$_{\dossmallformula{v}}$ that the length of the unique part of its \dosformula{indexing} is greater than 2, if \dosattack{} can collide out nodes \dosformula{nodes}$_{\dossmallformula{insert}}$ whose common prefix length with \dosformula{node}$_{\dossmallformula{v}}$ is at most \dosformula{x}, then \dosformula{node}$_{\dossmallformula{v}}$ can be deepened up to the layer \dosformula{x+2} (where \dosformula{x+2} is no larger than the maximum depth of \mpt{}) by inserting \dosformula{nodes}$_{\dossmallformula{insert}}$ in \mpt{}.
\hfill \qedsymbol
\end{lemma}

\begin{proof}[Proof of Lemma~\ref{lem_attack}]
Under node splitting triggered by \dosattack{}, the first \dosformula{x} nibbles of \dosformula{indexing} of \dosformula{node}$_{\dossmallformula{v}}$ will be maintained in \dosformula{x} intermediate nodes, because leaf nodes crafted by \dosattack{} can deepen \dosformula{node}$_{\dossmallformula{v}}$ (\textbf{S1-2}) and split the extension nodes whose maintained prefix length is larger than 1 (\textbf{S3}), and the length of pointers maintained in a branch node is 1 (\S\ref{sec_background_example}).
Hence, \dosformula{node}$_{\dossmallformula{v}}$ locates in the layer \dosformula{x+1} in \mpt{} (i.e., \dosformula{x} intermediate nodes are in front of \dosformula{node}$_{\dossmallformula{v}}$).
\dosattack{} can then deepen \dosformula{node}$_{\dossmallformula{v}}$ by 1 layer (i.e, deepening to the layer \dosformula{x+2}), with the length of the common prefix between the \dosformula{indexing} of \dosformula{node}$_{\dossmallformula{v}}$ and the inserted leaf node equaling to \dosformula{x} (\textbf{S1}).
\end{proof}

To trigger node splitting for a leaf node, \dosattack{} needs to craft leaf nodes, whose \dosformula{indexing} has a common prefix with the leaf node, and insert crafted leaf nodes in \mpt{} (\textbf{S1-3}).
However, it is challenging, because \dosformula{indexing} of the leaf node is derived from keccak256 hash computation, and hash algorithm is irreversibility~\cite{ramya2020securing} (\S\ref{sec_intro}).
To address the challenge, we design new methods to craft leaf nodes triggering node splitting by colliding the prefix of target leaf node's \dosformula{indexing}.
Specifically,
based on the parallel computing of GPUs,
we adopt OpenCL library~\cite{OpenCL,create2crunch} to collide the \dosformula{indexing} by parallelized computing keccak256 hash (Appendix~\ref{sec_collsion}).

\noindent
\textbf{Multi-target hash collision.}
There are multiple leaf nodes in \mpt{}, whose \dosformula{indexing} is required to be collided by \dosattack{}.
A trivial idea is to collide each leaf node's \dosformula{indexing} one by one.
We denote the counts of hash calculations to collide a specific \dosformula{indexing} as $\theta$.
Hence, the expected number of keccak256 hash calculations required for hash collision (denoted as \dosformula{E$_{\phi}$}) grows linearly with the number of leaf nodes whose \dosformula{indexing} is required to be collided (denoted as $\phi$), i.e., by the trivial method, \dosformula{E$_{\phi}$ = } $\theta$ $\times$ $\phi$.
Compared with the trivial idea, we propose a new multi-target hash collision strategy to collide all target leaf nodes' \dosformula{indexing} simultaneously, and by which, \dosformula{E$_{\phi}$} decreases to be as $\theta$ $\times$ \dosformula{ln(}$\phi$\dosformula{)}.

\begin{lemma}
\label{lem_multitarget}
Given $\phi$ leaf nodes whose \dosformula{indexing} is required to be collided, if \dosattack{} costs $\theta$ keccak256 hash calculations to collide a specific \dosformula{indexing}, \dosattack{} costs $\theta$ $\times$ \dosformula{ln(}$\phi$\dosformula{)} keccak256 hash calculation to collide all target \dosformula{indexing} simultaneously by multi-target hash collision.
\hfill \qedsymbol
\end{lemma}

\begin{proof}[Proof of Lemma~\ref{lem_multitarget}]
According to the multi-target collision search~\cite{quisquater1989easy, oechslin2003making} which investigates how many calculation counts are required to achieve one collision against multiple targets, \dosattack{} costs $\frac{\theta}{\phi}$ calculations to collide a \dosformula{indexing} of all target leaf nodes. 
In addition,
according to the coupon collector's problem~\cite{xu2011generalized}, \dosformula{E$_{\phi}$} grows with $\phi$ as the complexity of $\mathcal{O}(n \times ln(n))$~\cite{xu2011generalized}.
Hence, under our multi-target hash collision strategy, the required calculation counts for achieving hash collision against all target leaf nodes (\dosformula{E$_{\phi}$}) equals to the product of the calculation counts for colliding a \dosformula{indexing} of all target leaf nodes ($\frac{\theta}{\phi}$) and the complexity of how \dosformula{E$_{\phi}$} grows ($\phi \times \dosformula{ln(}\phi\dosformula{)}$).
Hence, \dosformula{E$_{\phi}$} can be derived by Eq.~\ref{eq_multitarget}.
\end{proof}

\begin{small}
\useshortskip
\vspace{3pt}
\begin{gather}
\label{eq_multitarget}
\dosformula{E$_{\phi}$} \dosformula{ = } \frac{\theta}{\phi} \times \phi \times \dosformula{ln(}\phi\dosformula{)}
\end{gather}
\useshortskip
\end{small}

\noindent
\textbf{Leaf node insertion}.
After obtaining the leaf nodes to be inserted in \mpt{}, \dosattack{} uses different methods to insert leaf nodes in State Trie \colorboxred{3} and Storage Tries \colorboxred{4}, respectively, e.g., transferring 1 wei Ether (i.e., \dosformula{$10^{-18}$} Ether) to a target EOA account.
We elaborate on the methods in the following.

\noindent
$\bullet$
Insert leaf nodes in State Trie \colorboxred{3}. The \dosformula{indexing} of a leaf node in State Trie \colorboxred{3} is derived from the address of an EOA. Hence, after the hash colliding, \dosattack{} will finally determine an EOA in such cases.
To insert the corresponding leaf node in \mpt{}, \dosattack{} directly sends 1 wei Ether (i.e., \dosformula{$10^{-18}$} Ether) to the EOA account by transactions.

\noindent
$\bullet$
Insert leaf nodes in Storage Tries \colorboxred{4}. Storage Tries \colorboxred{4} hold the persistent data for a contract's storage, and the \dosformula{indexing} of a leaf node in Storage Tries \colorboxred{4} is derived from the slot of the contract storage.
To insert a leaf node in Storage Tries \colorboxred{4}, \dosattack{} can only modify the data reserved in corresponding slot by interacting with the contract~\cite{wood2014ethereum}.
Please note that, for a key-value pair in a mapping~\cite{solidity}, e.g., \doscode{k} and \doscode{v}, \doscode{v} stores in a storage slot, and the slot is derived from \doscode{k} by keccak256 hash computation.
Hence, to insert a leaf node corresponding to a specific slot, \dosattack{} crafts elements (e.g. \doscode{k}) in mappings of contracts.
Specifically, after hash colliding, \dosattack{} determines the parameters~\cite{chen2021sigrec,zhao2023deep} to invoke a function in the target contract, which will update data in the target slot (corresponding to the target leaf node in Storage Tries \colorboxred{4}).
Taking ERC20 token contract~\cite{fabian2015eip20,he2023tokenaware,chen2019tokenscope} as an example, \dosattack{} finally determines the parameters for invoking \doscode{transfer()} by transferring the smallest unit of token (e.g., \dosformula{$10^{-6}$} USDT) to a target account address, where \doscode{transfer()} inserts the crafted leaf node in Storage Tries \colorboxred{4}, and set the data of the desired slot to be 1.

\noindent
\textbf{Optimization strategy}.
To launch \dosattack{}, its cost should be considered (\S\ref{sec_model}).
Expanding intermediate nodes associated with all accounts is impractical, because there are billions of leaf nodes in \mpt{}, and the corresponding cost is beyond the limited computing resources and assets of the adversaries in our model (\S\ref{sec_model}).
To reduce the cost, we propose an optimized strategy of \dosattack{} to only deepen the leaf nodes associated with active accounts.
Active accounts are the accounts that keep conducting frequent trades over a period of time~\cite{ethanos},
and they can be trivially captured by querying the frequency of each account being accessed and modified from blockchain.
According to the captured list, 
adversaries can strategically delineate the range of leaf nodes deepened by \dosattack{}, e.g., the leaf nodes associated with the accounts modified no less than six times in a specific range of time.
Since nodes associated with active accounts are keeping updating the reserved data, the resource consumption brought by \dosattack{} will be further exacerbated with limited cost.

\subsection{The implementation of \dosattack{}}
\label{sec_implement}

\begin{figure}[!t]
\small
	\centering
	\includegraphics[width=0.48\textwidth]{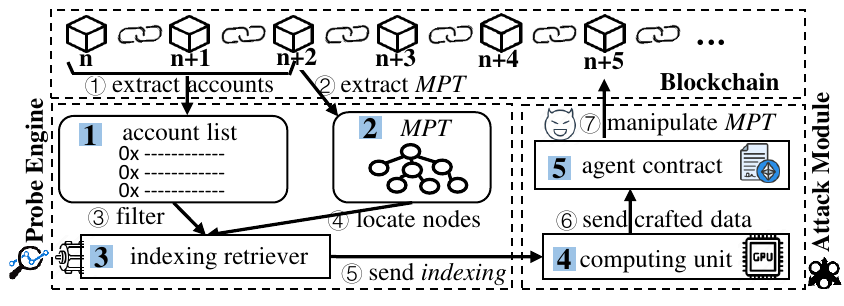}
	\caption{The implementation consists of seven steps. 
 Step \circlednum{1} and \circlednum{2}: \dosattack{} extracts the list of all accounts \protect\colorboxblue{1} and \mpt{} \protect\colorboxblue{2} from blockchain. 
 Step \circlednum{3}: The indexing retriever \protect\colorboxblue{3} retrieves target accounts from the list of all accounts \protect\colorboxblue{1}. 
 Step \circlednum{4}: The indexing retriever 
 \protect\colorboxblue{3} locates target leaf nodes in \mpt{} \protect\colorboxblue{2}, and derives \dosformula{indexing} of them.
 Step \circlednum{5}: The indexing retriever \protect\colorboxblue{3} forwards derived \dosformula{indexing} to the computing unit \protect\colorboxblue{4} for colliding \dosformula{indexing}. 
 Step \circlednum{6}: The computing unit \protect\colorboxblue{4} generates the crafted data (e.g., addresses of accounts) from \dosformula{indexing} with the desired prefix, and sends the crafted data to agent contract \protect\colorboxblue{5}. 
 Step \circlednum{7}: The agent contract \protect\colorboxblue{5} 
  manipulates \mpt{} by inserting crafted leaf nodes.} 
	\label{fig_implement}
\end{figure}

Fig.~\ref{fig_implement} shows the overview of \dosattack{}'s implementation.
There are seven steps in \dosattack{}
encompassing three core portions: i) Blockchain serves as the data source and attack target. 
ii) Probe engine analyzes the accounts information and \mpt{}, and retrieves the \dosformula{indexing} of target leaf nodes.
iii) Attack module wields computing resources to generate the crafted data from the retrieved \dosformula{indexing}, and invokes the agent contract to insert the crafted leaf nodes into \mpt{}. We portray the implementation of \dosattack{} by steps \circlednum{1} to \circlednum{7}.

\noindent\textbf{Blockchain.} In Step \circlednum{1} and \circlednum{2}, \dosattack{} collects the list of all accounts~\colorboxblue{1} and \mpt{}~\colorboxblue{2} from blockchain. The account list \colorboxblue{1} is used to capture active accounts. 
Besides, we record the frequency of each account being accessed and modified in the account list \colorboxblue{1}. \mpt{} \colorboxblue{2} is used to retrieve the \dosformula{indexing} of active accounts' leaf nodes for proliferating intermediate nodes. We instrument logic of \circlednum{1} and \circlednum{2} in blockchain client, and run the client in real-time for data collection.

\noindent\textbf{Probe engine.} Probe engine retrieves \dosformula{indexing} of target leaf nodes by analyzing the list of accounts~\colorboxblue{1} and \mpt{}~\colorboxblue{2},
and \dosattack{} aims to trigger node splitting and proliferate intermediate nodes associated with the leaf nodes.
In Step \circlednum{3}, Probe engine first retrieves active accounts from the account list~\colorboxblue{1}. We determine an active account when the frequency of it being modified and accessed is larger than a threshold (e.g., 6)
in a specific block range~\cite{ethanos}. In Step \circlednum{4}, the indexing retriever~\colorboxblue{3} retrieves the leaf nodes and their \dosformula{indexing} corresponding to active addresses from \mpt{} (\S \ref{sec_design}).
In step \circlednum{5}, the indexing retriever~\colorboxblue{3} sends obtained \dosformula{indexing} to the Attack module for the computing of hash collision.

\begin{figure}[!b]
\scriptsize
	\centering
    \begin{lstlisting}[mathescape=true,language=Solidity, frame=none, basicstyle=\scriptsize \ttfamily]
contract NurglePrototype{
  function NurgleState(address payable[] memory Payloads) payable public{
    uint256 len = Payloads.length;
    for(uint256 i=0; i < len; i++){
      bool result=Payloads[i].send(1);}}
  function NurgleStorage(address dst, bytes4 funcsig, bytes[] memory Payloads, uint num) public { 
    uint256 len = Payloads.length/num;
    for(uint256 i=0; i < len; i++){
      bytes memory encode=abi.encodePacked(funcsig);
      for(uint256 j=num*i;j <num*i+num;j++){
        encode=abi.encodePacked(encode, Payloads[j]);}
      dst.call(encode);}}} 
    \end{lstlisting}
	\caption{Code snippet of agent contract. \doscode{NurgleState()} (Line 2-5) inserts the crafted leaf nodes into State Trie \protect\colorboxred{3} by sending 1 wei Ether to target accounts. \doscode{NurgleStorage()} (Line 6-12) inserts crafted leaf nodes on Storage Tries \protect\colorboxred{5}. As the logic of how smart contracts access their Storage Tries \protect\colorboxred{5} can be distinct, \doscode{NurgleStorage()} allows the adversary to customize: i) \doscode{dst}, the callee contact, ii) \doscode{funcsig}, the function to be invoked, and iii) \doscode{Payloads}, the parameters for invoking the function. In Line 9-11, \doscode{NurgleStorage()} splices \doscode{funcsig} and \doscode{Payloads}. In Line 15, \doscode{NurgleStorage()} invokes \doscode{dst} with crafted payload, which executes the logic of \doscode{dst} to insert the leaf nodes on Storage Tries \protect\colorboxred{5}.}
	\label{fig_agent}
\end{figure}

\noindent\textbf{Attack module.} The Attack module generates the corresponding crafted data through the computing unit~\colorboxblue{4}, and then inserts the crafted leaf nodes into \mpt{} through the agent contract~\colorboxblue{5}. Computing unit~\colorboxblue{4} utilizes GPU resources for hash computing. The agent contract~\colorboxblue{5} is deployed by the adversary, and Fig~\ref{fig_agent} displays the code snippet of the agent contract. 
In step \circlednum{6}, \dosattack{} leverages the computing unit~\colorboxblue{4} to collide \dosformula{indexing} with desired prefix to generate the crafted data (e.g., the address of an account). Computing unit~\colorboxblue{4} then sends the crafted data to the agent contract~\colorboxblue{5}.
In Step \circlednum{7}: The agent contracts~\colorboxblue{5} insert crafted leaf nodes into \mpt{} by invoking its functions (i.e., \doscode{NurgleState()} and \doscode{NurgleStorage()}). \doscode{NurgleState()} (Line 2-5) inserts crafted leaf nodes in State Trie \colorboxred{3} by sending 1 wei Ether to target accounts. \doscode{NurgleStorage()} (Line 6-12) inserts crafted leaf nodes in Storage Tries \colorboxred{5} by invoking specific functions of target contracts (e.g., \doscode{transfer()} function of ERC20 token contract) with crafted payload.

\section{Evaluation}
\label{sec_evaluaion}

We answer four research questions for evaluating \dosattack{}'s cost and impact. 
\textbf{RQ1:} How do computing resources influence the attack impact of \dosattack{}?
\textbf{RQ2:} How severe is the attack impact of \dosattack{} on the current blockchain? 
\textbf{RQ3:} How is the economic feasibility of \dosattack{}? 
\textbf{RQ4}: Will \dosattack{} threaten the current blockchain in practice?

\noindent\textbf{Experimental setup.} 
We evaluate \dosattack{} on a server with an Intel Xeon Gold 5218R CPU (2.10 GHz, 10 cores), 64 GB RAM, 1 TB SSD, and single RTX3080 GPU. 
We adopt a go-ethereum client at v1.11.6~\cite{gethclient} to measure the consumed resources of blockchain. 
We evaluate the impact of \dosattack{} on blockchain by the time cost of state modification, 
because it can comprehensively reflect the consumed resources, e.g., CPU computation, and the load and read for memory and disk, during state modification~\cite{ponnapalli2021rainblock}.
Please note that we do not explicitly distinguish modifying and maintaining state, since they are interwoven in \textbf{OP1-4}.

\begin{table*}[]
\centering
\caption{Time cost for $\dosattack{}$ to collide different lengths of desired prefix for an \dosformula{indexing}.}
\label{tab_preiamges}
\resizebox{0.99\linewidth}{!}{
\begin{tabular}{@{}l|l|l|c@{}}
\toprule
\begin{tabular}[c]{@{}l@{}}Digits\end{tabular} & Crafted data (hex encoding) & $\dosformula{indexing}$ (hex encoding) & Time cost  \\
\midrule \midrule
1 & 0x51b0e4b84afc9c7e935fd1c54409abda46ffff07 & 0x$\underline{1}$09999afd60b733da226a060260c2d9f165f0f33516c5a3230d2b9538ae197e7 & $\dosformula{$<$1s}$  \\
2 & 0x7c0caee5b72d0c71a090c6f02522e89acfffff07 & 0x$\underline{11}$fb9e6a64c5a7c23fb27d08e3d74ea1018fcb0c60d2010cca6c6654dd95e4b8 & $\dosformula{$<$1s}$  \\
3 & 0x8f5ea3c9db43de4143e5717f44dcb43e05d0fe07 & 0x$\underline{111}$0dc62b86ce4609e860381909da5480d46b2e90ea19c5afac287be805c234b & $\dosformula{$<$1s}$  \\
4 & 0xbd6f8cba28b4a0218d0aedbc820a27248ee4fe07 & 0x$\underline{1111}$65e10752633a1ab85c219c618d6c6e6259fdb7c8d2397df9cb72d16e4149 & $\dosformula{$<$1s}$  \\
5 & 0xfccedcfd14858e8b1baf9a497e99af468012b507 & 0x$\underline{11111}$0e0c5d11a713c428c42a03a5a7c55d66c0e61158ef13a63776b94d384d0 & $\dosformula{$<$1s}$  \\
6 & 0x58b91f9cb0ffacae5d95c9e80c373d264993cc06 & 0x$\underline{111111}$078c719cdc5abc2195b645a72ba7dd4d15b12ab9cce3361466c402df69 & $\dosformula{$<$1s}$ \\
7 & 0x89f25e63c12c48a95c22cd4b19585f337a805f06 & 0x$\underline{1111111}$06b6090ca5f7027a7539dc73173e26a35b28645b47d4878db6bbddd62 & $\dosformula{$<$1s}$ \\
8 & 0xa0f0722109f07edd76cc1d2b29cfbc0122ca2b06 & 0x$\underline{11111111}$ce35790ede4c97cc847e55c91c0b3063f5cb56ab6ab93ee76381fa6a & $\dosformula{$<$1s}$ \\
9 & 0x97637e992f835689667a48a0731ce1ebb44dc006 & 0x$\underline{111111111}$0b3bf4ed6dc409fb20328970a0f23dac93761a4347fcd4c84dfe8cc & $\dosformula{53s}$  \\
10 & 0x2f1033b78f8fb3c04259202793d2d89169326d02 & 0x$\underline{1111111111}$ce8bad4529bfef324c88454fe4e72c3cd3974c0249c9adc764802a & $\dosformula{21.68m}$  \\
11 & 0x267a239f1986295e996358a79f57b473ae264d05 & 0x$\underline{11111111111}$00822f67e0319be36eb814ade0ca60c65c62b41641e889eb48ad8 & $\dosformula{2.8h}$\\
12 & 0xd4dfd776a81fcdfa2d601f1efa31a2ad8c21fe06 & 0x$\underline{111111111111}$834eea3006374356f398b29f9b709272533e759348f0bb07aa11 & $\dosformula{12.57h}$ \\
13 & 0xdf04b72b67666a59ff30c06dd079f1850b36ba04 & 0x$\underline{1111111111111}$ca536d3de683a3ab986f631ee733132457eccc0d9a011aa9e55 & $\dosformula{24.58h}$ \\ \bottomrule
\end{tabular}
}
\vspace{3pt}
\end{table*}

\subsection{How do computing resources affect \dosattack{}?}
\label{sec_evaluate_cap}

\dosattack{} crafts leaf nodes that contain a common prefix with a target leaf node, and inserts the crafted leaf nodes to deepen the target leaf node and proliferate intermediate nodes, causing extra consumed resources in modifying and maintaining \mpt{} (\textbf{OP1-4}) (\S\ref{sec_design}).
Besides, according to Lemma~\ref{lem_attack}, for a leaf node, if \dosattack{} can craft another leaf node that contains a common prefix with the target node at the length of \dosformula{x}, and then \dosattack{} can deepen the target leaf node to the layer \dosformula{x+2}.
Hence, the larger \dosformula{x} that \dosattack{} can find out, the deeper the target node can be deepened (i.e., \dosformula{x+2}).
Therefore, to assess capability of \dosattack{} under the different cost of computing resources, we evaluate the required computing resources of the adversary for colliding the different lengths of the common prefix.

\ignore{
\begin{small}
\begin{gather}
\label{eq_power}
\dosformula{T}=\frac{\dosformula{16}^\dossmallformula{n}}{\dosformula{P}\times\dosformula{3600}}
\end{gather}
\end{small}

Eq.~\ref{eq_power} relies on a practical assumption that the adversary adopts a brute force strategy to collide the target prefix (cf. details in Appendix~\ref{sec_collsion})~\cite{kelsey2006herding}.
In Eq.~\ref{eq_power},
we denote
\dosformula{T} as the required time to collide a desired prefix of the target \dosformula{indexing}, \dosformula{n} as the length of the desired prefix to be collided, \dosformula{16}$^\dossmallformula{n}$ as the search space of the hash collision, and \dosformula{P} as the number of hash collisions that an adversary can perform per second.
}

To uniform the comparison of consumed computing resources, we fix hardwares used to conduct \dosattack{} (e.g., single RTX3080 GPU), and estimate the required computing resources by utilizing the cost time for \dosattack{} to collide a desired prefix at different lengths.
In our evaluation, we launch \dosattack{} to collide an \dosformula{indexing}, i.e., \doscode{0x1111..1111}, for demonstration.
Besides, our evaluation relies on a practical assumption that the adversary adopts a brute force strategy to collide the target prefix (cf. details in Appendix~\ref{sec_collsion})~\cite{kelsey2006herding}.
Under the fixed computing resources (e.g., single RTX3080 GPU), \dosattack{} can conduct 1.90 billion hash calculations per second, 
and we record the time cost until \dosattack{} successfully crafts the target prefix.
Table~\ref{tab_preiamges} shows our experimental results.
In Table~\ref{tab_preiamges}, the first column lists the length of the desired prefix in hash collision, the second column lists the final data crafted by \dosattack{}, the third column lists the corresponding \dosformula{indexing} derived by the crafted data in the second column, and the fourth column lists the time cost for \dosattack{} to collide the desired prefix.
As a result, under a reasonable consumption of computing resources, the adversary can collide out a desired prefix at the length of 13, which will deepen the target node to the layer 15 (Lemma~\ref{lem_attack} in \S\ref{sec_design}).
Please note that, the adversary can shorten the time cost of \dosattack{} by just deploying more GPUs or switching more powerful GPUs.

\noindent\textbf{Answer to RQ1}:
\textit{
Computing resources affect the length of the collided prefix for a target \dosformula{indexing} crafted by \dosattack{}.
Under a reasonable consumption of computing resources, \dosattack{} can craft a desired prefix at the length of 13, which can deepen the target node to the layer 15 in \mpt{}.
}

\subsection{How does \dosattack{} threaten blockchains?}
\label{sec_evaluate_mainnet}

\noindent
\textbf{Estimating attack impact}.
Time cost of blockchain for modifying and maintaining its \mpt{} (i.e., \textbf{OP1-4} in Table~\ref{tab_stateops}) linearly increases with the number of nodes in \mpt{} involved in modifying and maintaining state.
Hence, we can estimate the overhead raised by \dosattack{} in state modification and maintenance according to the number of nodes in \mpt{} proliferated by \dosattack{}.
Specifically,
we assume that \dosattack{} deepens several leaf nodes, and
Eq.~\ref{eq_grossnode} derives the overhead (i.e., \dosformula{F}$_{\dossmallformula{nurgle}}$) in updating and maintaining deepened leaf nodes in \mpt{} brought by \dosattack{}.
\dosformula{F}$_{\dossmallformula{nurgle}}$ is obtained by dividing i) the number of nodes involved in handling (i.e., updating and maintaining) deepened leaf nodes in \mpt{} under the attack by ii) the number of nodes involved in handling deepened leaf nodes in \mpt{} without the attack. 
In Eq~\ref{eq_grossnode}, \dosformula{Num}$_{\dossmallformula{StateTrie}}$ and \dosformula{Num}$_{\dossmallformula{StorageTries}}$
are denoted as the number of nodes involved in handling deepened leaf nodes in State Trie \colorboxred{3} and Storage Tries \colorboxred{4} without the attack, respectively.
Besides, \dosformula{Num}$_{\dossmallformula{StateTrie}}^{\dossmallformula{'}}$ and \dosformula{Num}$_{\dossmallformula{StorageTries}}^{\dossmallformula{'}}$ are denoted as the number of corresponding nodes in State Trie \colorboxred{3} and Storage Tries \colorboxred{4} under the attack, respectively.
Please note that, 
nodes in \mpt{} can be partitioned into two parts, i.e., nodes in State Trie \colorboxred{3} and nodes in Storage Tries \colorboxred{4} (\S\ref{sec_background_state}).
Therefore, 
the number of nodes in \mpt{} involved in handling deepened leaf nodes in \mpt{} without and under the attack of \dosattack{} are \dosformula{Num}$_{\dossmallformula{StateTrie}}$\dosformula{+}\dosformula{Num}$_{\dossmallformula{StorageTries}}$ and 
\dosformula{Num}$_{\dossmallformula{StateTrie}}^{\dossmallformula{'}}$\dosformula{+}\dosformula{Num}$_{\dossmallformula{StorageTries}}^{\dossmallformula{'}}$, respectively.

\begin{small}
\useshortskip
\vspace{3pt}
\begin{gather}
\label{eq_grossnode}
\dosformula{F}_{\dossmallformula{nurgle}}\dosformula{ = }\frac{\dosformula{Num}_{\dossmallformula{StateTrie}}^{\dossmallformula{'}}\dosformula{+}\dosformula{Num}_{\dossmallformula{StorageTries}}^{\dossmallformula{'}}}{\dosformula{Num}_{\dossmallformula{StateTrie}}\dosformula{+}\dosformula{Num}_{\dossmallformula{StorageTries}}} \\
\label{eq_statenode}
\dosformula{Num}_{\dossmallformula{StateTrie}}^{\dossmallformula{'}}\dosformula{ = } 
\dosformula{Num}_{\dossmallformula{StateTrie}}\dosformula{+}\dosformula{Num}_{\dossmallformula{Account}}\times\dosformula{(}\dosformula{d}_{\dossmallformula{nurgle}} \dosformula{-}\dosformula{d}_{\dossmallformula{base}}\dosformula{)}  \\
\label{eq_contractnode}
\dosformula{Num}_{\dossmallformula{StorageTries}}^{\dossmallformula{'}} \dosformula{=}
\dosformula{Num}_{\dossmallformula{StorageTries}}\dosformula{+}\dosformula{Num}_{\dossmallformula{Slot}}\times\dosformula{(}\dosformula{d}_{\dossmallformula{nurgle}}\dosformula{-}\dosformula{d}_{\dossmallformula{base}}\dosformula{)}
\end{gather}
\end{small}

\begin{figure}[!t]
\small
	\centering
	\includegraphics[width=0.4\textwidth]{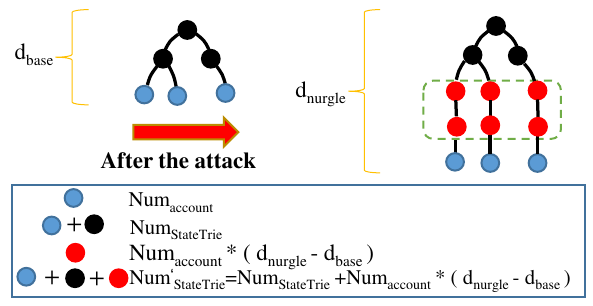}
	\caption{ The nodes in \mpt{} without and under the attack. 
 }
	\label{MPT_DOS_formula}
\end{figure}

The number of nodes involved in handling deepened leaf nodes in \mpt{} can be trivially obtained by traversing \mpt{}.
In addition,
we also present models to estimate \dosformula{Num}$_{\dossmallformula{StateTrie}}^{\dossmallformula{'}}$ (Eq.~\ref{eq_statenode}) and \dosformula{Num}$_{\dossmallformula{StorageTries}}^{\dossmallformula{'}}$ (Eq.~\ref{eq_contractnode}),
which can be used by adversaries to estimate the attack impact by Eq.~\ref{eq_grossnode} before launching \dosattack{}.
Specifically, the number of nodes involved in handling deepened leaf nodes under the attack is the sum of i) the number of corresponding nodes without the attack, and ii) the number of proliferated intermediate nodes involved in handling deepened leaf nodes.
Besides, the above proliferated intermediate nodes can be derived by the product of i) the number of deepened leaf nodes (i.e., \dosformula{Num}$_{\dossmallformula{Account}}$ in State Trie \colorboxred{3} and \dosformula{Num}$_{\dossmallformula{Slot}}$ in Storage Tries \colorboxred{3}), and ii) the number of layers that the leaf nodes are deepened (i.e., \dosformula{d}$_{\dossmallformula{nurgle}}$\dosformula{-}\dosformula{d}$_{\dossmallformula{base}}$).
In Eq.~\ref{eq_statenode} and Eq.~\ref{eq_contractnode}, \dosformula{d}$_{\dossmallformula{nurgle}}$ 
and \dosformula{d}$_{\dossmallformula{base}}$ refer to the layers where corresponding leaf nodes are located after and without the attack, respectively.
For Fig.~\ref{MPT_DOS_formula} as an example, three leaf nodes are involved in state modification in the State Trie \colorboxred{3}, and \dosattack{} deepens them by two layers.
Without the attack, \dosformula{Num}$_{\dossmallformula{StateTrie}}$ equals to 6.
Besides, \dosformula{Num}$_{\dossmallformula{Account}}$ equals to 3, and \dosformula{d}$_{\dossmallformula{nurgle}}$\dosformula{-}\dosformula{d}$_{\dossmallformula{base}}$ equals to 2.
Hence, according to Eq.~\ref{eq_statenode}, \dosformula{Num}$_{\dossmallformula{StateTrie}}^{\dossmallformula{'}}$ equals to 6 $+$ 3 $\times$ 2, i.e., 12.
It means that \dosattack{} causes 200\% overhead (12/6) for maintaining and updating the three leaf nodes.

\ignore{
\begin{table}[!b]\centering
\caption{Six parameters captured from Ethereum}
\label{tab_parameters}
\resizebox{0.99\linewidth}{!}{
\begin{tabular}{@{}c|c|c|c|c|c@{}}
\toprule
 \dossmallformula{Num}$_{\dostinyformula{Account}}$  & \dossmallformula{Num}$_{\dostinyformula{StateTrie}}$ & \dossmallformula{Num}$_{\dostinyformula{Slot}}$  & \dossmallformula{Num}$_{\dostinyformula{StorageTries}}$ & \dossmallformula{d}$_{\dostinyformula{base}}$ & \dossmallformula{d}$_{\dostinyformula{nurgle}}$ \\

 \midrule \midrule
 712,565 & 7,116,002 & 2,041,719 & 6,506,806  & 9.5  & 15     \\
\bottomrule
\end{tabular}
}
\end{table}
}

\noindent
\textbf{Evaluating attack impact}.
We evaluate attack impact in the fork of Ethereum, because it is the most popular blockchain under threats of \dosattack{}.
Specifically,
we synchronize an instrumented client~\cite{gethclient} to capture the transactions of Ethereum from the block height of \#14.99M to \#15M (i.e., 10,000 blocks).
By replaying captured transactions in corresponding blockchain state~\cite{zhou2021a2mm,li2023blockexplorer,chen2019dataether}, we retrieve the nodes in \mpt{} involved in state modification in the transactions.
Specifically, we retrieve 7,116,002 nodes from State Trie \colorboxred{3} and 6,506,806 nodes from Storage Tries \colorboxred{4}, which contain 712,565 and 2,041,719 leaf nodes, respectively.
We then launch \dosattack{} in our forked \mpt{} of Ethereum at the block height of \#15M to deepen the retrieved leaf nodes to the layer 15.
Finally, we traverse \mpt{} to count how many intermediate nodes are proliferated by \dosattack{}, which will be involved in the state modification of previously retrieved leaf nodes.
As a result, we find that the number of proliferated intermediate nodes is 1.11 times of all the retrieved nodes.
It means that, under the attack, the time of state
modification from \#14.99M to \#15M consumes more 111\% resources, because it requires handling more 111\% nodes proliferated by \dosattack{}.
Besides, for all future transactions involving the nodes added by \dosattack{}, the time cost of their state modification is persistently raised.
Please note that the majority of blockchain execution time is consumed by interactions with \mpt{} (\S\ref{sec_observation_stateop}).
Hence, the increase in the number of nodes involved in \mpt{} can raise a considerable overhead on the execution performance of the blockchain, leading to performance degradation in the blockchain.
\ignore{As mentioned in~\S\ref{sec_observation_stateop}, the operations interacting with \mpt{} cost over 81\% execution time of blockchain. Hence, 
we can roughly estimate that \dosattack{} raises more than 89.91\% overhead of its execution performance, indicating that \dosattack{} can significantly degrade the execution performance of blockchain.}

\ignore{
\begin{table}[!b]\centering
\caption{Six parameters captured from Ethereum}
\label{tab_parameters}
\resizebox{0.99\linewidth}{!}{
\begin{tabular}{@{}c|c|c|c|c|c@{}}
\toprule
 \dossmallformula{Num}$_{\dostinyformula{Account}}$  & \dossmallformula{Num}$_{\dostinyformula{StateTrie}}$ & \dossmallformula{Num}$_{\dostinyformula{Slot}}$  & \dossmallformula{Num}$_{\dostinyformula{StorageTries}}$ & \dossmallformula{d}$_{\dostinyformula{base}}$ & \dossmallformula{d}$_{\dostinyformula{nurgle}}$ \\

 \midrule \midrule
 712,565 & 7,116,002 & 2,041,719 & 6,506,806  & 9.5  & 15     \\
\bottomrule
\end{tabular}
}
\end{table}
}

We further validate whether we can successfully estimate the attack impact.
Specifically, the values of the six parameters (i.e., \dosformula{Num}$_{\dossmallformula{Account}}$, \dosformula{Num}$_{\dossmallformula{StateTrie}}$,
\dosformula{Num}$_{\dossmallformula{Slot}}$,
\dosformula{Num}$_{\dossmallformula{StorageTries}}$
\dosformula{d}$_{\dossmallformula{base}}$, and
\dosformula{d}$_{\dossmallformula{nurgle}}$) from our captured transactions are 712,565, 7,116,002, 2,041,719, 6,506,806, 9.5, and 15, respectively.
To simplify the estimation of attack impact, we have averaged out the \dosformula{d}$_{\dossmallformula{base}}$.
Based on the six parameters, we derive that the corresponding \dosformula{F}$_{\dossmallformula{Nurgle}}$ equals to 2.112.
Our models estimate that there are more 111.2\% nodes proliferated by \dosattack{}, causing 111.2\% more consumed resources in handing \mpt{}.
Our results validate that our models estimate the impact of the attack almost perfectly (difference $<$ 1\%).

\noindent\textbf{Answer to RQ2}:
\textit{
\dosattack{} significantly degrades the execution performance of blockchain.
}

\begin{table*}[!t]
\centering
\caption{The cost of \dosattack{} on seven different blockchains}
\label{tab_cost}
\resizebox{0.99\linewidth}{!}{
\begin{threeparttable}
\begin{tabular}{@{}l||c|c|c|c|c||c|c|c@{}}
\toprule
Blockchain  & \dosformula{Price}$_\dossmallformula{coin}$ (USD) & \dosformula{Price}${_\dossmallformula{gas}}$ & \makecell{\dosformula{G}${_\dossmallformula{gas}}$ (USD)} &
\makecell{\dosformula{G}${_\dossmallformula{gpu}}$ (USD)} & \makecell{\dosformula{G}${_\dossmallformula{nurgle}}$ (USD)} &
\makecell{\begin{tabular}[c]{@{}c@{}}Optimized\\ \dosformula{G}${_\dossmallformula{gas}}$ (USD)\end{tabular}}&
\makecell{\begin{tabular}[c]{@{}c@{}}Optimized\\ \dosformula{G}${_\dossmallformula{gpu}}$ (USD)\end{tabular}} & 
\makecell{\begin{tabular}[c]{@{}c@{}}Optimized\\ \dosformula{G}${_\dossmallformula{nurgle}}$ (USD)\end{tabular}}
 \\
 \midrule \midrule
Ethereum &	1,812	& 22.5 G wei	&  11,808,917.46 	& 39.6 	& 11,808,957.06 & 413,312.11 & 33
  & 413,345.11\\
\begin{tabular}[c]{@{}l@{}}Binance Smart Chain\end{tabular}	& 252.71	& 2.71 G wei	& 198,360.95 &	39.6 &	198,400.55  & 6,942.63 & 33& 6,975.63\\
Heco	&2.81	&2.5 G wei	& 2,034.77 	& 39.6 	& 2,074.37 & 71.21 & 33 & 104.21\\
Polygon&	0.71	&206.30 G wei	& 42,596.22 
 & 39.6 &	42,635.82 & 1,490.86 & 33 & 1,523.86\\
Optimism	&1,812	& 9.35$\times$10$^{-8}$ G wei & 0.049   &	39.6  &	39.649 & 0.0017 & 33  & 33.0017 \\
Avalanche	&22.66	&27.76 n AVAX &182,200.15 	& 39.6 	& 182,239.75 & 6,377.00 & 33 & 6,410.00\\
\begin{tabular}[c]{@{}l@{}}Ethereum Classic\end{tabular}	&16.52	&1.17 G wei	&5,596.96 	& 39.6 &	5,636.56 & 195.89 & 33 & 228.89 \\
\bottomrule
\end{tabular}
\vspace{-1pt}
\begin{tablenotes}[flushleft]
{
\setlength{\itemindent}{-2.49997pt} \small
\item  \footnotesize{Due to the volatility of cryptocurrency and gas prices, we have calculated their average values over a one-week period, spanning from June 6, 2023, to June 12, 2023.}
}
\end{tablenotes}
\end{threeparttable}
}
\vspace{3pt}
\end{table*}

\subsection{How much does \dosattack{} cost?}
\label{sec_evaluate_eco}

The cost of leveraging \dosattack{} is an essential metric for a financially rational adversary.
According to~\S\ref{sec_model}, the cost of \dosattack{} (denoted as \dosformula{G}$_{\dossmallformula{nurgle}}$) consists of two parts, i.e., \dosformula{G}$_{\dossmallformula{gpu}}$, the cost of computing resources (mainly GPUs~\cite{liu2022ready}) for the calculation of hash collision (\S\ref{sec_design}), and \dosformula{G}$_{\dossmallformula{gas}}$, the cost of gas fee for submitting attack payloads to blockchain via transactions.
Besides, 
\dosformula{G}$_{\dossmallformula{gpu}}$ is derived by the product of i) \dosformula{Num}$_{\dossmallformula{gpu}}$, the number of GPUs utilized by adversaries, ii) \dosformula{Time}$_{\dossmallformula{hours}}$, the hours of renting GPUs from GPU markets~\cite{gpumarket} by adversaries, and iii) \dosformula{Price}$_{\dossmallformula{gpu}}$, the USD price for renting a GPU in GPU markets.
Please note that, when multiple GPUs are required in launching \dosattack{}, it is reasonable for adversaries to minimize the cost by renting GPUs from GPU markets (e.g.,~\cite{gpumarket}) for a short period of time~\cite{liu2022ready}.
Furthermore, 
according to the specifications of gas mechanism~\cite{wood2014ethereum},
\dosformula{G}$_{\dossmallformula{gas}}$ is derived by the product of 
i) \dosformula{Price}$_{\dossmallformula{gas}}$, the cryptocurrency (e.g., Ether) price of a unit of gas, 
ii) \dosformula{Units}$_{\dossmallformula{gas}}$, the units of consumed gas for executing the transactions containing attack payloads, and iii) \dosformula{Price}$_{\dossmallformula{coin}}$, the USD price of the cryptocurrency.
We present corresponding equations for assessing the USD cost of \dosattack{} in Eq.~\ref{eq_cost0} - Eq.~\ref{eq_cost1}.

In the rest of this section, we estimate the cost of \dosattack{} on seven mainstream blockchains (\S\ref{sec_rq3_seven}), and evaluate how the cost of \dosattack{} can be optimized (\S\ref{sec_rq3_active}).

\begin{small}
\useshortskip
\vspace{3pt}
\begin{gather}
\label{eq_cost0}
\dosformula{G}_{\dossmallformula{nurgle}}\dosformula{ = }\dosformula{G}_{\dossmallformula{gas}}+\dosformula{G}_{\dossmallformula{gpu}} \\
\label{eq_cost2}
\dosformula{G}_{\dossmallformula{gpu}}\dosformula{ = }\dosformula{Num}_{\dossmallformula{gpu}} \times \dosformula{Time}_{\dossmallformula{hours}} \times \dosformula{Price}_{\dossmallformula{gpu}}\\
\label{eq_cost1}
\dosformula{G}_{\dossmallformula{gas}}\dosformula{ = }\dosformula{Price}{\dossmallformula{gas}}\times \dosformula{Units}_{\dossmallformula{gas}} \times \dosformula{Price}_{\dossmallformula{coin}}
\end{gather}
\end{small}

\ignore{
Eq.~\ref{eq_cost0} to \ref{eq_cost2} formalize the cost per block of \dosattack{}. 
Eq.\ref{eq_cost0} formalizes the gross cost of \dosattack{}, and the gross cost  \dosformula{G}${_\dossmallformula{nurgle}}$ consists of two portions, the gas fee  \dosformula{G}${_\dossmallformula{gas}}$ and the Computing resources fee \dosformula{G}${_\dossmallformula{gpu}}$. 

First, we derive the gas fee \dosformula{G}${_\dossmallformula{gas}}$ in Eq.\ref{eq_cost1}. 
The gas fee is mainly consumed on transactions initiated by \dosformula{}, such as creating new addresses and inserting leaf nodes into MPT.
The \dosformula{G}${_\dossmallformula{gas}}$ is multiplied by three parameters, namely, the units of gas per block consumed by \dosattack{} (i.e., \dosformula{Units}${_\dossmallformula{gas}}$), the gas price (i.e., \dosformula{Price}${_\dossmallformula{Gas}}$), and the coin price of the target blockchain (i.e., \dosformula{Price}$_\dossmallformula{coin}$). 
Note that different blockchains have divergent gas prices and coin prices.

Then, we obtain the Computing resources fee \dosformula{G}${_\dossmallformula{gpu}}$ by Eq~\ref{eq_cost2}. 
We consider that the adversary adopts leased mode~\cite{gpu_lease} to obtain the computing unit(i.e., GPUs), as the expense on leased mode can include energy expense and make the evaluation more accurate and practical. 
Concretely, the \dosformula{G}${_\dossmallformula{gpu}}$ is first multiplied by three parameters, namely, the number of GPU (i.e., \dosformula{Num}$_{\dossmallformula{gpu}}$), rental time (i.e., \dosformula{Time}$_{\dossmallformula{hours}}$), and rental price (i.e., \dosformula{Price}$_{\dossmallformula{gpu}}$) , and then divided by the range of blocks (\dosformula{Block}${_\dossmallformula{range}}$).  
}

\subsubsection{Cost for attacking seven blockchains}
\label{sec_rq3_seven}
We demonstrate \dosattack{}'s attack towards
seven popular blockchains (i.e., Ethereum, BSC, Heco, Polygon, Optimism, Avalance, and Ethereum Classic)~\cite{coinmarketcap}, and measure corresponding attack cost of \dosattack{}. 
Since node numbers and layers in \mpt{} of the seven blockchains are distinct, according to \S\ref{sec_evaluate_cap}, the 
attack impact of \dosattack{} is different on the seven blockchains.
To uniformly and fairly compare the cost of \dosattack{} on different blockchains, we fix the attack impact of \dosattack{} on different blockchains.
Specifically,
we reuse the attack impact for the attack launched by us in \S\ref{sec_evaluate_mainnet} on each blockchain to measure the cost of launching \dosattack{}.

We further derive actual values for parameters in Eq.~\ref{eq_cost0} - Eq.~\ref{eq_cost1} on seven blockchains in the following.
Since the attack impact is fixed, the corresponding attack procedure of \dosattack{} should also be fixed, e.g., the procedure of hash collision (\S\ref{sec_design}).
Hence, for the same attack impact on different blockchains, the cost of computing resources (i.e., \dosformula{G}$_{\dossmallformula{gpu}}$) is the same. 
Specifically,
according to \S\ref{sec_evaluate_mainnet},
there are 2,754,284 (712,565 + 2,041,719) leaf nodes in \mpt{} to be collided for being deepened by \dosattack{}.
To conduct the hash collision for the leaf nodes, according to Lemma~\ref{lem_multitarget} and experimental results in \S\ref{sec_evaluate_cap}, adversaries need to rent 33 RTX3080 GPUs for a period of 12 hours at least (cf. details in Appendix~\ref{sec_app_hashtimecost}).
We obtain the corresponding price of GPU rental from~\cite{gpumarket}, i.e., 0.1 USD/hour for renting a GPU.
Therefore, to launch the attack in \S\ref{sec_evaluate_mainnet}, \dosformula{G}$_{\dossmallformula{gpu}}$ is 39.6 USD (0.1 $\times$ 33 $\times$ 12).
Furthermore, 
since the seven blockchains are compatible with Ethereum and adopt the same gas mechanism,
the units of gas consumed for executing the transactions containing attack payloads are also the same.
Specifically,
during the attack in \S\ref{sec_evaluate_mainnet},
after forwarding attack payloads to the agent contract \colorboxblue{5}, the agent contract \colorboxblue{5} executes its logic to insert crafted leaf nodes (\S\ref{sec_design}) deepening the target 7,425,484 leaf nodes to the layer 15.
As a result, it costs 289,647,227,381 units of gas.
We further acquire the price of cryptocurrency (i.e., \dosformula{Price}$_{\dossmallformula{coin}}$) and the cryptocurrency price of a unit of gas (\dosformula{Price}$_{\dossmallformula{gas}}$) of each blockchain from corresponding dashboards (e.g.,~\cite{coinmarketcap,gasfees1}).
Since \dosformula{Price}$_{\dossmallformula{coin}}$ and \dosformula{Price}$_{\dossmallformula{gas}}$ are volatile over time, we average out them in a period of one week.

Based on the derived parameters in Eq.~\ref{eq_cost0} - Eq.~\ref{eq_cost1},
we list the detailed cost of \dosattack{} on different blockchain in Table~\ref{tab_cost}.
The second and third columns list the price of cryptocurrency and the cryptocurrency price of a unit of gas for each blockchain.
For example, \dosformula{Price}$_{\dossmallformula{coin}}$ of Optimism is 1,812 USD, and \dosformula{Price}$_{\dossmallformula{gas}}$ of Optimism is 9.35$\times$10$^{-8}$.
The fourth, fifth, and sixth columns list \dosformula{G}$_{\dossmallformula{gas}}$, \dosformula{G}$_{\dossmallformula{gpu}}$, and \dosformula{G}$_{\dosattack{}}$ for launching \dosattack{} on different blockchain.
For example, to degrade the performance of Optimism to 47\% of original performance for a period of 10,000 blocks (\S\ref{sec_evaluate_mainnet}), it only costs 39.64 USD for launching \dosattack{}.

\ignore{
Next, we seek to measure the cost of \dosattack{} in seven blockchain platforms with the highest market capitalization, including Ethereum, BSC, Heco, Polygon, Optimism, Avalance, and ETC by Eq.~\ref{eq_cost0} to \ref{eq_cost2}. 
To simplify the problem, we consider applying the actual stats on Ethereum to appraise the attack cost of the remaining platforms. 
These stats mainly relate to the number of leaves that \dosattack{} needs to attack
By analyzing the volume of daily transaction volume~\cite{tx_eth, tx_etc,tx_polygon,tx_avalanche,tx_heco,tx_bsc,tx_optimistic} and daily trading accounts~\cite{addrs_avalanche,addrs_bsc,addrs_etc,addrs_polygon,addrs_eth}, we find that the stats of Ethereum can be used to assess the cost of caps on five platforms other than BSC. 
For BSC, by analyzing the daily transaction volume~\cite{tx_bsc,tx_eth} and daily trading address~\cite{addrs_bsc,addrs_eth}, we believe that the scale of data of BSC is now about four times that of Ethereum. 

Therefore, we conduct the cost of \dosattack{} based on the stats from Ethereum \#14.99M to \#15M blocks. We assume that the attacker's goal is to achieve impacts as described in \S \ref{sec_evaluate_mainnet}.
Table~\ref{tab_cost} enumerates the cost of \dosattack{} per block for seven blockchains platforms. 

First, we interpret how to obtain gas fee \dosformula{G}${_\dossmallformula{gas}}$ by Eq.\ref{eq_cost1}. The \dosformula{Price}$_\dossmallformula{coin}$ and \dosformula{Price}${_\dossmallformula{gas}}$ is the mean between Jun-6-2023 to Jun-12-2023.
\dosformula{Units}${_\dossmallformula{gas}}$ indicates the units of gas required per block by \dosattack{}. 
To obtain the \dosformula{Units}${_\dossmallformula{gas}}$, we conduct statistics on the state modified of accounts and slots of contract storage in Ethereum between \#14.99M to \#15M blocks.
We observe 3,258,711 changes for 712,565 accounts and 4,166,773 changes for 2,041,719 storage slots of contracts. 
Therefore, from \#14.99M to \#15M, we can acquire that \dosattack{} requires attack 71 accounts and 204 slots of contract storage per block. To attack them, we need 28,928,988 (for BSC is 115,715,952) units of gas per block.(cf. details in Appendix \ref{sec_gasguage}). 
Finally, we can acquire the gas fee per block by \dosattack{} in Table~\ref{tab_cost}. 
The gas fee of Ethereum is as high as 1179USD and the gas fee of Optimism is as low as 5.02$\times$10$^{-6}$ USD. 
Note that the active account-based attack strategy can further optimize the costs.

We derive GPUs expense by Eq~\ref{eq_cost2}. The GPUs are leased from vast.ai~\cite{gpu_lease}, the market leader in on-demand GPUs rental. 
We need 33 RTX3080 GPUs for 12 hours' rental and for BSC we need 13 hours' rental (cf. details in Appendix \ref{sec_app_hashtimecost}) and the price of a RTX3080 GPU is 0.1USD/hour on vast.ai. 
The range of blocks is 10,000 blocks (i.e., from \#14.99 million to \#15 million blocks). Therefore, through the Eq~\ref{eq_cost2}, we can obtain the \dosformula{G}${_\dossmallformula{gpu}}$ equal to 4$\times$10$^{-3}$ USD (for BSC is 4.2$\times$10$^{-3}$ USD). With the GPU fee , we can obtain the final costs of the attack. 

In Table.~\ref{tab_cost}, except for Ethereum, the sums of \dosattack{} per block of all other blockchains is below 80USD.
}

\subsubsection{Cost optimization} 
\label{sec_rq3_active}
According to Table~\ref{tab_cost}, the high gas fee leads to the expensive cost, e.g., the cost is over 11M USD on Ethereum. 
Inspired by active accounts (\S\ref{sec_design}), we adopt an optimized strategy to decrease the cost of \dosattack{} by only deepening leaf nodes associated with active accounts, costing less and achieving a trade-off between attack impact and cost of \dosattack{}.
Active accounts are accounts conducting frequent
trades over a period of time (\S\ref{sec_design}), 
hence, the leaf nodes associated with them are the most frequently modified and accessed leaf nodes in \mpt{}.

We further inspect the 2,754,284 leaf nodes deepened by \dosattack{} during the attack in~\S\ref{sec_evaluate_mainnet}.
It shows that 2,103,558 of them (i.e., 361,703 and 1,741,855 leaf nodes in State Trie \colorboxred{3} and  Storage Tries \colorboxred{4}) are only modified and accessed once among transactions of the 10,000 blocks (\S\ref{sec_evaluate_mainnet}).
Compared with the leaf nodes,
the other 650,726 leaf nodes (2,754,284-2,103,558) have been collectively modified and accessed 5,321,926 times.
Hence, if \dosattack{} only deepens the 23.63\% (650,726/2,754,284) leaf nodes, the attack impact of \dosattack{} retains 71.67\% (5,321,926/(2,103,558+5,321,926)) of original attack impact.
Besides, \dosattack{}'s \dosformula{G}$_{\dossmallformula{gas}}$ is only 19.61\% of original \dosformula{G}$_{\dossmallformula{gas}}$ (\S\ref{sec_observation_stateop}).
Please note that, the cost of deepening leaf nodes in State Trie \colorboxred{3} and Storage Tries \colorboxred{4} is different (Appendix~\ref{sec_gasguage}).
We enumerate other cases in Table~\ref{tab_optimized_cost} for launching \dosattack{} on Ethereum.
For example, if \dosattack{} only deepens the leaf nodes modified and accessed no less than 6 times, \dosformula{G}$_{\dossmallformula{gas}}$ of \dosattack{} is 3.5\% of original \dosformula{G}$_{\dossmallformula{gas}}$, and the retained impact is 54.66\% of original attack impact.
As a result, we provide the optimized cost (e.g., \dosformula{G}$_{\dossmallformula{gas}}$) of \dosattack{} on the seven blockchains in Table~\ref{tab_cost} in the cases that \dosattack{} only deepens the leaf nodes associated with the active accounts modified and assessed no less than 6 times.
Since the number of leaf nodes deepened by \dosattack{} decreases, the corresponding \dosformula{G}$_{\dossmallformula{gpu}}$ also reduces to 33 USD accordingly (Appendix~\ref{sec_app_hashtimecost}).
It shows that, by optimized strategies based on active accounts, the cost of \dosattack{} on Polygon reduces from 42,596.22 USD to 1,523.86 USD, retaining 54.66\% of original attack impact.

\begin{table}[!t]
\caption{Cost optimization based on active accounts}
\label{tab_optimized_cost}

\resizebox{0.99\linewidth}{!}{
\begin{tabular}{@{}l|c|c|c@{}}
\toprule
Count & Retained impact
& Retained cost
& \makecell{Optimized \dosformula{G}$_{\dossmallformula{gas}}$ (USD)}
\\
 \midrule \midrule
 1&100.00\% & 100.00\% & 11,808,917.46 \\
 2& 71.67\% & 19.60\% & 2,314,547.82\\
 4 & 58.16\% & 5.76\% & 680,193.64 \\
 6 & 54.66\% & 3.50\% & 413,312.11\\
\bottomrule
\end{tabular}
}
\end{table}

\noindent
\textbf{Baseline comparison.} We compare the cost of creating spam transactions that result in the same attack impact as \dosattack{} with optimized strategies.
Specifically, for each block within the range of \#14.99M to \#15M, we iteratively submitted transactions sampled from Ethereum to the block, ensuring that the number of \mpt{} nodes handled in the submitted transactions aligned with the number of \mpt{} nodes proliferated by \dosattack{}.
As a result, the baseline costs 3,826,037.45 USD on Ethereum, which is 9.25 times higher than \dosattack{} with optimized strategies in Table~\ref{tab_cost}. 

\noindent\textbf{Answer to RQ3}:
\textit{
\dosattack{}'s cost depends on the gas fee, and is reasonable for most cases, and our optimization strategy further minimizes the cost of \dosattack{}.
}

\subsection{Can \dosattack{} threaten blockchain in practice?}
\label{sec_evaluate_testnet}

Compared to demonstrating \dosattack{} in controlled environments~\cite{Li2023LVMT}, we choose to determine whether \dosattack{} can threaten the current blockchain in practice. 
Please note that, unlike previous studies~\cite{ndss_broken_metro,heopartitioning,li2021deter} that only have non-persistent attack impact, the impact of \dosattack{} will be persistent in the blockchain. This is because \dosattack{} persistently proliferates the intermediate nodes in \mpt{}, and persistently exacerbates resource consumption for the maintaining and updating of \mpt{}.
Therefore, due to ethical concerns and inspired by previous studies~\cite{li2021deter,heopartitioning}, we choose to launch \dosattack{} on blockchain testnets, i.e., Ethereum Sepolia testnet~\cite{sepolia_testnet} and BSC testnet~\cite{bsc_testnet}.
Please note that testnets set the closest environment to the practice, and it is built for researchers and developers to conduct experiments without risk to real funds or the main chain~\cite{Ethereum_testnet}.
Considering that there are other developers and researchers who are active in the testnet, we further minimize the potential ethical issues by carefully adjusting attack parameters from scratch to light the attack impact on testnets.

In our evaluation, we synchronize blockchain clients to obtain latest state of the two testnets.
We launch \dosattack{} to exploit the two testnets by following the seven steps of \dosattack{} (\S\ref{sec_implement}).
Specifically, \dosattack{} first extracts the list of all accounts \colorboxblue{1} in the two testnets and \mpt{} of two testnets \colorboxblue{2}.
We next leverage indexing retriever \colorboxblue{3} to retrieve the target leaf nodes and their \dosformula{indexing}.
After computing unit \colorboxblue{4} crafts the leaf nodes to be inserted for proliferating the intermediate nodes, we forward the crafted data to our agent
contracts (i.e., \href{https://sepolia.etherscan.io/address/0xc8f2b352a53ef7f1d5c33f4bc8129d50ae1f199d}{0xc8f2...199d} in Ethereum testnet and \href{https://testnet.bscscan.com/address/0xc06256021c13ea3bdc6053d6422df2c4df60a163}{0xc062...a163} in BSC testnet) \colorboxblue{5} to 
finalize the attack.

\begin{figure}[t]
\small
  \begin{subfigure}{\linewidth}
    \centering
	\includegraphics[width=0.99\linewidth]{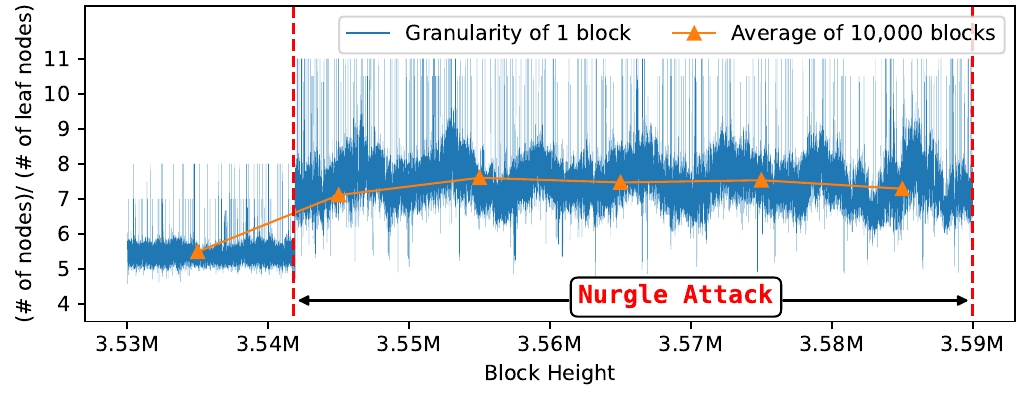}
	\caption{The tendency of involved MPT nodes of Ethereum testnet.}
    \label{fig_uniswapexample_a}
  \end{subfigure}

  \begin{subfigure}{\linewidth}
    \centering
	\includegraphics[width=0.99\linewidth]{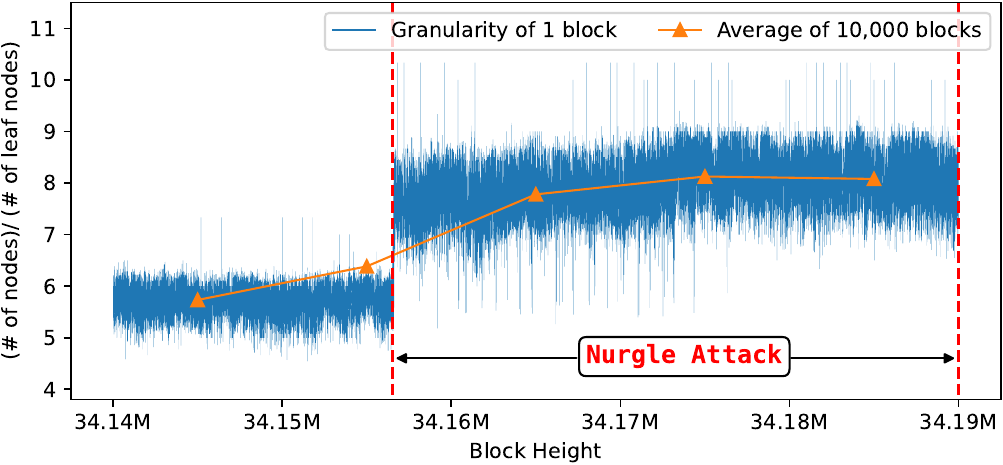}
	\caption{The tendency of involved MPT nodes of BSC testnet.}
    \label{fig_smoothyexample_b}
  \end{subfigure}
  \caption{
During the exploitation, \dosattack{} proliferates intermediate nodes in \mpt{}, and deepens leaf nodes, leading to an increase in the number of \mpt{} nodes to be modified to update a leaf node.
\dosattack{} can cause about 32\% and 39.4\% more \mpt{} nodes to be involved in state modification per block in Ethereum testnet and BSC testnet, respectively.
}
  \label{exp_mpt_node}
\end{figure}

\ignore{
\begin{figure}[!t]
\small
	\centering
	\includegraphics[width=0.49\textwidth]{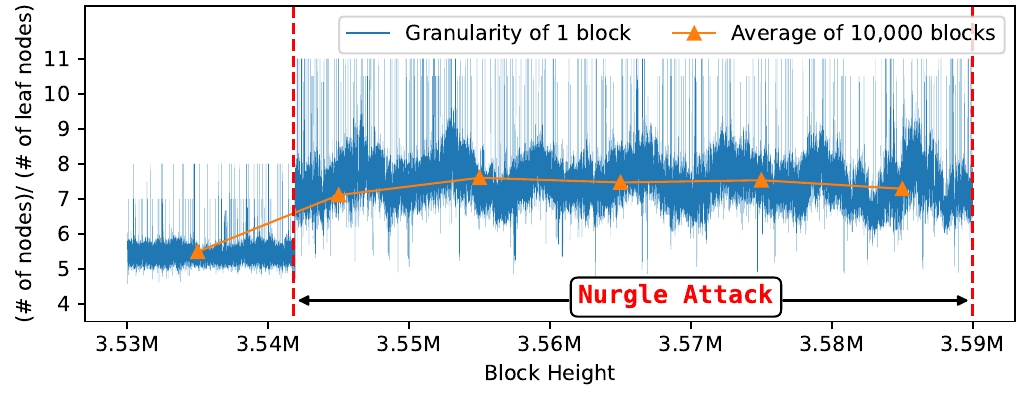}
	\caption{The tendency of nodes in \mpt{} involved
in state modification.
During the exploitation, \dosattack{} proliferates intermediate nodes in \mpt{}, and deepens leaf nodes, leading to an increase in the number of nodes in \mpt{} to be modified to update a leaf node.
\dosattack{} can cause about 32\% more nodes in \mpt{} to be involved in state modification per block.}
	\label{exp_mpt_node}
\end{figure}
}

We launched the attack of \dosattack{} on the Ethereum (resp. BSC) testnet at the block height of \#3,541,798 (May 23, 2023) (resp. \#34,156,452 (Oct. 13, 2023)), and ceased the attack at the block height of \#3.59M (Jun. 1, 2023) (resp. \#34.17M (Oct. 13, 2023)).
During the whole exploitation of \dosattack{}, we forwarded attack payloads to agent
contracts \colorboxblue{5} by 330 (resp. 53) transactions, and we
inserted 100 leaf nodes into State Trie \colorboxred{3} in \mpt{} for each transaction. 
Fig.~\ref{exp_mpt_node} shows our experimental results, and it depicts the tendency of the number of nodes in \mpt{} to be modified to update a leaf node. 
Specifically,
during the exploitation, \dosattack{} proliferates intermediate nodes in \mpt{}, and deepens leaf nodes, leading to an increase in the number of nodes in \mpt{} to be modified to update a leaf node.
As mentioned in~\S\ref{sec_observation_stateop}, the consumed resources (e.g., the time cost) for state modification (e.g., \textbf{OP1-4}) linearly increases with the number of involved nodes.
Hence,
Fig.~\ref{exp_mpt_node} indicates that the state modification is significantly exacerbated by the exploitation of \dosattack{}.
To obtain a comprehensive understanding of the impact of \dosattack{}'s exploitation on the testnet, we further investigate the performance overhead of the blockchain brought by \dosattack{}.
As a result, during the attack against Ethereum (resp. BSC) testnet, 32\% (resp. 39.4\%) more nodes in \mpt{} are involved in state modification per block, besides, \dosattack{} raises the cost time of state modification by 15\% (resp. 18\%). 
We further evaluate the overall performance degradation caused by \dosattack{} by using the Metrics module~\cite{gethmetrics}, which enables us to collect execution information about blockchain clients.
As a result, \dosattack{} results in a 10.7\% (resp. 12.4\%) rise in end-to-end execution time of Ethereum (resp. BSC) testnet.

\noindent\textbf{Answer to RQ4}:
\textit{
\dosattack{} threatens blockchain by causing more nodes involved in state modification, and raises the cost time of the overall blockchain execution.
}

\section{Discussion}
\label{sec_discuss}
\subsection{Practical attack impact of \dosattack{}}
\label{sec_dis_impact}
\dosattack{} can threaten blockchain in seven aspects.

\noindent
i) As the ever-evolving blockchains~\cite{Ethereum_eco}, attack scenarios of \dosattack{} are extremely rich, especially the emerging blockchains.
Taking a newly deployed blockchain as an example, and assuming that its leaf nodes in \mpt{} are in the layer 5, \dosattack{} can deepen a leaf node to the layer 15 (\S\ref{sec_evaluate_cap}), thereby proliferating extra 10 intermediate leaves for the blockchain to updating the leaf node.

\noindent
ii) The impact of \dosattack{} on blockchain is persistent, and the victim blockchain will be impacted by  manipulated \mpt{} in all subsequent blocks by the nodes proliferated by \dosattack{}.
Besides, according to~\S\ref{sec_evaluate_cap}, the attack impact of \dosattack{} can be further exacerbated by adversaries with more powerful computing resources (e.g., GPUs).

\noindent
iii) \dosattack{} delays users in using blockchain and AUX (e.g., flashbot~\cite{weintraub2022flash,li2023demystifying}, infura~\cite{li2021strong}, ENS~\cite{XiaENSIMC22}) in providing services, because \dosattack{} exacerbates resource consumption of blockchain, and increases the time cost for maintaining and updating \mpt{}.
For example, 
an AUX like infura can only provide
their services after they finish the delay of updating the latest
state in \mpt{}.

\noindent
iv) Since \dosattack{} can delay the processing of user transactions, attackers can launch the delay attacks~\cite{delayattack} to threaten the liveness of the layer 2 rollups. Specifically, attackers can utilize \dosattack{} to slow down the confirmation of the transactions for verifying the validity of layer 2 transactions. 

\noindent
v) \dosattack{} threatens the consensus security of blockchain by increasing the execution costs of running blockchain nodes, which subsequently results in reducing the number of nodes participating in the blockchain network~\cite{heopartitioning}.

\noindent
vi) \dosattack{} erodes trust in the affected blockchains, leading to a decline in the value of their cryptocurrencies~\cite{ndss_broken_metro}.

\noindent
vii) The overhead of blockchain execution raised by \dosattack{} (e.g., the 10.7\% overall performance degradation in \S\ref{sec_evaluate_testnet}) can waste the energy of all blockchain nodes.

As mentioned in~\S\ref{sec_evaluate_cap},
computing resources (i.e., GPUs) of adversaries influence the attack impact of \dosattack{}. Concretely, computing resources affect the depth of target leaf nodes in \mpt{} deepened by \dosattack{}, and the number of intermediate nodes proliferated by \dosattack{}, which finally impacts how much the resource consumption in \mpt{} will be exacerbated to impair the blockchain’s performance.
Our evaluation on \dosattack{} suffers some limitations that we do not explore the best attack impact of \dosattack{}, because we choose to evaluate the attack impact of \dosattack{} under a reasonable resource cost of \dosattack{} (\S\ref{sec_evaluate_cap}).
Hence, our experimental results should be considered as the lower bound of the potential attack impact brought by \dosattack{}.

\subsection{On-demand attacks of \dosattack{}}
\label{sec_on_demand}

Instead of one-time short exploitation,
it is also feasible to control \dosattack{} on demand. 
During the on-demand attacks, attackers can craft several deeper leaf nodes controlled by themselves in advance, and only update it to slow down transactions when needed.
We further evaluate the impact and cost of on-demand attacks by examining two distinct on-demand attacks, each targeting users of different contracts. Our experimental results demonstrate that on-demand attackers can slow down user transactions of specified contracts by crafting attack payloads in advance. Besides, while implementing attacks with the same impact (e.g., slowing down all user transactions), the attack cost depends on the logic of involved contracts.
In the following, we elaborate on how we evaluate the two on-demand attacks and the detailed experimental results.
 
\noindent
$\bullet$
In the first attack, the adversary slows down all users of an AMM (Automated Market Maker) contract~\cite{xu2021sok}, where users can exchange two specific tokens with the AMM contract. Since the token balances of the AMM contract will update when users interact with it, the adversary can deepen the two leaf nodes storing the token balances of the AMM contract for the two tokens to slow down all user transactions. The cost of transaction fees is 4.06 USD on Ethereum, and the cost of computing resources is 4.32 USD.

\noindent
$\bullet$
In the second attack, the adversary delays all users of a token contract. Since, for each transaction, only the token balances of the users in the current transaction will update, the adversary needs to deepen all the leaf nodes storing all users' token balances to slow down all user transactions. In our experiments, the token contract has 10,000 users. Hence, the cost of transaction fees is 20,154.95 USD on Ethereum, and the cost of computing resources is 23.56 USD.

\subsection{Feasible mitigations against \dosattack{}}
\label{sec_mitigations}
In the following, we detail the three feasible mitigations, and discuss their advantages and disadvantages.

\noindent
i) \textbf{Verkle tree}~\cite{kuszmaul2019verkle}. Verkle tree mitigates the impact of \dosattack{} by indexing fewer nodes in its tree structure and adopting a swifter authenticated method~\cite{kuszmaul2019verkle}.
First, compared with the prefix tree in MPT structure, the structure of verkle tree is designed to be flatter, which compresses the distance from the leaf node to the root node. The verkle tree alleviates the volume of nodes to be updated and verified associated with \textbf{OP1} and \textbf{OP2} (\S\ref{sec_observation_stateop}). 
Second, verkle tree adopts a polynomial commitment scheme~\cite{boneh2021halo} rather than the hash-style vector commitment of MPT structure~\cite{wood2014ethereum}.
During the verification for the latest state, the polynomial commitment reduces the number of nodes to be verified~\cite{boneh2021halo}. 
Verkle tree has the advantage of consuming fewer resources when accessing a leaf node by involving fewer intermediate nodes. However, the advantage comes at the cost of increased space wastage in Verkle tree, as each node consumes more space. For instance, a branch node in Verkle tree, which includes 256 pointers, occupies nearly 16 times more space than a branch node in \mpt{}, which only has 16 pointers.

\noindent
ii) \textbf{Trimming historical state (EIP-4444)}~\cite{eip4444}. EIP-4444 eases \dosattack{}'s impact by periodically reducing \mpt{}'s size. 
Specifically, \mpt{} will undergo pruning by retaining the state of blockchain for nearly one year. Ethanos~\cite{ethanos} has carried out its implementation in a similar manner. 
Trimming historical state has the advantage of reducing the storage space required by a blockchain node by eliminating a portion of its historical state. However, the mitigation comes with the disadvantage that the blockchain node cannot conduct complete security verification of the whole blockchain state due to the elimination of the historical state.

\noindent
iii) \textbf{New patch for gas mechanism}. We plan to propose a new Ethereum improvement proposal (EIP) for the gas mechanism to defend against \dosattack{}. In our EIP, gas fee of transactions will be equivalent to actual consumed resources, e.g., considering the consumed resources for modifying nodes in \mpt{}.
The mitigation has the advantage of implementing a fairer gas mechanism. The adjusted mechanism ensures that the gas fee for transactions corresponds directly to the actual resources consumed when modifying nodes in the MPT. However, the mitigation can introduce new attack vectors. For instance, an adversary can exploit the adjusted mechanism by proliferating intermediate nodes to popular contracts, thereby increasing gas fees for all users of the contracts.
Hence, it is crucial to examine the implementation of the third mitigation thoroughly.

\subsection{Vulnerability disclosure}
\label{sec_vul_disclose}
At the time of writing, 588 blockchains compatible with the Ethereum ecosystem (\S\ref{sec_background}) are directly under the threats of \dosattack{}. 
Besides, we reveal that
other 153 blockchains compatible with the Polkadot ecosystem~\cite{Polkadot_eco} are also  threatened by \dosattack{}, 
because Polkadot also adopts the MPT structure to handle state storage similar to Ethereum (\S\ref{sec_background}). However, there are two major differences between Polkadot and Ethereum. i) Hash algorithm.
Polkadot utilizes xxhash algorithm~\cite{collet2016xxhash} instead of keccak256 algorithm~\cite{keccak256} in their blockchain design, e.g., they use xxhash to derive the \dosformula{indexing} of its leaf nodes. 
ii) Fee mechanism. Polkadot adopts the weight-based mechanism to cost fee for transactions rather than gas mechanism adopted in Ethereum.
Concretely, transaction fee is adjusted according to the congestion cost of the block. 
We enumerate the above vulnerable blockchains under \dosattack{} in Appendix~\ref{sec_scope}.
We will explore the hash collision of \dosattack{} (\S\ref{sec_design}) against xxhash algorithm and optimize attack parameters against the fee mechanism of Polkadot to further
investigate the vulnerability of Polkadot under \dosattack{} as our future work.

\ignore{
\begin{table}[]\centering
\caption{Response status of seven mainstream blockchains}
  \label{tab_responce}
\begin{tabular}{@{}l|l@{}}
\toprule
\begin{tabular}[c]{@{}l@{}}Blockchain platform\end{tabular} & Response status\\
 \midrule
Ethereum  & Acknowledged the vulnerabilities. \\ \midrule
\begin{tabular}[c]{@{}l@{}}
Binance Smart Chain
\end{tabular}   & \begin{tabular}[c]{@{}l@{}}Acknowledged with the vulnerabilities, and \\rewarded us with bug bounty.\end{tabular} \\ \midrule
Heco      & No reponse     \\ \midrule
Polygon   & Acknowledged the vulnerabilities. \\ \midrule
Optimism  & Acknowledged the vulnerabilities. \\ \midrule
Avalanche & \begin{tabular}[c]{@{}l@{}}Acknowledged with the vulnerabilities, and \\rewarded us with bug bounty.\end{tabular}\\ \midrule
\begin{tabular}[c]{@{}l@{}}
Ethereum\\ Classic
\end{tabular}       & Acknowledged the vulnerabilities.  \\
\bottomrule
\end{tabular}
\end{table}
}

\ignore{
\noindent\textbf{Feasible mitigations against \dosattack{}} 
There are three feasible mitigations. 
i) Verkle tree~\cite{kuszmaul2019verkle}. Verkle tree mitigates the impact of \dosattack{} by indexing fewer nodes in its tree structure and adopting a swifter authenticated method~\cite{kuszmaul2019verkle}.
First, compared with the prefix tree in MPT structure, the structure of verkle tree is designed to be flatter, which compresses the distance from the leaf node to the root node. The verkle tree alleviates the volume of nodes to be updated and verified associated with \textbf{OP1} and \textbf{OP2} (\S\ref{sec_observation_stateop}). 
Second, verkle tree adopts a polynomial commitment scheme~\cite{boneh2021halo} rather than the hash-style vector commitment of MPT structure~\cite{wood2014ethereum}.
During the verification for the latest state, the polynomial commitment reduces the number of nodes to be verified~\cite{boneh2021halo}. 
ii) Trimming historical state (EIP-4444~\cite{eip4444}). EIP-4444 eases the persistent impact of \dosattack{} by periodically reducing the size of \mpt{}. 
Specifically, the \mpt{} will undergo pruning by retaining the state of the blockchain for nearly one year. Ethanos~\cite{ethanos} has carried out its implementation in a similar manner. 
iii) Design a new patch for gas mechanism. We plan to propose a new Ethereum improvement proposal (EIP) for the gas mechanism to defend against \dosattack{}. In our EIP, gas fee of transactions will be equivalent to actual consumed resources, e.g., considering the consumed resources for modifying nodes in \mpt{}.
}

\noindent\textbf{Ethics concerns.} We reported vulnerabilities brought by \dosattack{} to seven mainstream blockchains with high market capitalization before the paper submission. 
At the time of writing, we have received responses from six of them, confirming the vulnerabilities and rewarding us with thousands of USD bounty.
Furthermore, before the publication, we reported the vulnerabilities to all other affected blockchains.
Currently, an additional 16 teams from the newly reported blockchain teams have responded with positive acknowledgments.
We present the details of their responses in Appendix~\ref{sec_app_status}.
Besides, we will release corresponding materials 90 days before the publication following ethical obligations and conference committee requirements.

\ignore{
\hzy{\noindent\textbf{The affecting factors of \dosattack{}.} i) Target Blockchain. Coin prices and gas fee vary from blockchain to blockchain, and this difference can greatly affect the cost of an attack. ii) The data on Blockchain. Data on the blockchain, such as the depth of the MPT or the frequency of access to the account, can determine the consequences of \dosattack{}. iii) Computing resources. The computational resources at the disposal of the adversary can determine the feasibility of \dosattack{}, i.e., how many intermediate nodes to increase. we depict a general attack algorithm considering such factors in Algorithm \ref{alg_attack}.}  
}

\ignore{
\subsection{Threats to validity}
\noindent\textbf{The limitation of \dosattack{}.} While \dosattack{} has the notable advantages mentioned above, it also has some limitations. 

i) The trade-off between Computer resources and attack consequences. Computer resources (i.e., GPUs) constrain the capability of \dosattack{}. Concretely, computational resources affect the depth of MPT that \dosattack{} can manipulate, which finally impacts attack consequences. We do not adopt the best trade-off to maximize the impact of \dosattack{}, but we think that the attack impact can be further exacerbated. 

ii) Gas fee. The adversary needs to be motivated to pay the gas fee to launch the attack. In our threat model, it is worthwhile for the adversaries to pay the cost to launch \dosattack{}. For instance, the adversaries, as a speculator or competitor, makes profits by reducing the efficiency of the target operation and hitting the confidence of the target's users~\cite{ndss_broken_metro}. In some cases, adversaries can exploit project airdrops~\cite{harrigan2018airdrops} to launch \dosattack{} for free. Besides, the Gas fee is a common requirement in these related dos attacks~\cite{ndss_broken_metro,li2021deter}, so the gas fee for \dosattack{} is also reasonable.

\noindent\textbf{The limitation of evaluation.} Our evaluation in \S \ref{sec_evaluate_cap} did not conduct the experiment on Storage Tries \colorboxred{4}. Unlike the State Trie\colorboxred{3}, the key on the Storage Tries \colorboxred{4} usually performs two keccak256 operations. Concretely, the EVM (Ethereum virtual machine) might first perform a keccak256 hash (i.e., SHA3 instruction) for the key, and then perform the second keccak256 hash calculations when inserting the Trie. Although there will also be contracts performing multiple SHA3 instructions in the EVM, we observe that most popular contracts~\cite{popular_con} only perform SHA3 instructions once. The double keccak256 operations are not enough to weaken the feasibility of \dosattack{}, as it can only halve the computing power.
}

\section{Related work}
\label{sec_related}
In this section, we explore closely related studies in four aspects by the attack surfaces of DoS attacks on blockchain.

\noindent\textbf{Consensus network.} 
Based on the fault-tolerant mechanism, consensus network
assists blockchain nodes in achieving an agreement on the latest state. Compromising consensus network incurs the blockchain to violate its consensus functionalities~\cite{chen2022tyr}.
Heo et al.~\cite{heopartitioning} achieve the isolation and disconnection of an Ethereum node from main network by hijacking half of peer connections on the blockchain network. Tran et al.~\cite{tran2020stealthier} propose an attack that isolates Bitcoin nodes via malicious Internet Service Providers. Chen et al.~\cite{chen2022tyr} detect vulnerabilities that cause denial of service in the consensus network through fuzz testing~\cite{peng2018t}.
Prunster et al.~\cite{prunster2022total} mount eclipse attack~\cite{heilman2015eclipse} on Inter Planetary File System (IPFS) by poisoning the node's routing information, so that the node is isolated from the main network. 
Yang et al.~\cite{fluffyosdi} unveil a vulnerability that can collapse the consensus of the blockchain. The root cause of the vulnerability is that the state of the blockchain client implementation in different languages is inconsistent. 
Saad et al.~\cite{saad2021syncattack} propose SyncAttack which uses fluctuations in the Bitcoin network to achieve blockchain forks. 
SDoS~\cite{wang2022sdos} is a DoS attack based on selfish mining. They find that an adversary can launch a 51\% attack~\cite{badertscher2021rational} to destroy the consensus of the blockchain with only 19.6\% of the computing power.

\noindent\textbf{Txpools.} 
Txpools maintain pending transactions from blockchain users, and miners/validators will pack transactions from their txpools to blockchain.
Adversaries sway the security of blockchain by interfering with the transaction packing involved in txpools.
Deter~\cite{li2021deter} paralyzes Ethereum transaction pool by crafting malicious transactions, and prevents users from interacting with the blockchain. 
Wu et al.\cite{wu2020survive} propose a distributed denial-of-service (DDoS) attack against the Bitcoin mining pool with the idea of game theory.
Poster\cite{saad2018poster} denies the memory pool of Bitcoin and traps users into paying higher mining fees. 
Yaish et al.~\cite{yaish2023speculative} leverage the censorship mechanism adopted in Ethereum to craft three classes of DoS attacks, resulting in burdening victims’ computational resources and clogging their txpools.

\noindent\textbf{Auxiliary services.}
Auxiliary services refer to entities facilitating blockchain's efficiency at off-chain.
DoS attackers can refuse users to utilize services provided by AUX.
Li et al.~\cite{li2021strong} propose a DoS attack against RPC services in Ethereum.
It indirectly causes users to be unable to access the Ethereum mainnet by paralyzing the RPC service. 
Nguyen et al.~\cite{NguyenTc2023} discover a flood attack against a single shard, which undermines the performance of blockchain.
They present a scheme by the Trusted Execution Environment (TEE) to counter the flood attack.

\noindent\textbf{Smart contracts.} 
Smart contracts are automation programs executed in blockchain. One of their security issues 
is
under-priced opcodes~\cite{chen2017adaptive}, whose gas cost is substantially less than the consumed resources. 
Attackers can lead blockchain to consume extremely high resources while executing under-priced opcodes.
Chen et al.~\cite{chen2017adaptive} reveal the under-priced opcodes on Ethereum, and attackers can wield these opcodes to launch DoS attacks.
They propose an adaptive gas mechanism to defend against the DoS attacks mentioned by them by auto-adjusting the gas fee of opcodes.
Perez et al.~\cite{ndss_broken_metro} use genetic methods~\cite{whitley1994genetic} to generate smart contracts with low gas and high resource consumption, by leveraging the under-priced instructions of Ethereum. eTainter~\cite{ghaleb2022etainter,grech2020madmax} inspects the DoS attack based on the gas mechanism of smart contracts through static program analysis. 
In addition, the design of \dosattack{} is also motivated by existing studies that exploit historical state~\cite{li2021strong, ndss_broken_metro,chen2017adaptive}.

\section{Conclusion}
\label{sec_conclusion}

We reveal a new attack surface, i.e., state storage, of blockchains.
Besides, we present \dosattack{},
the first DoS attack exploiting state storage. By proliferating intermediate nodes within state storage, \dosattack{} forces blockchains to expend additional resources on
state maintenance and verification, impairing their performance.
We further conduct a comprehensive and systematic evaluation of \dosattack{}.
Experimental results show that \dosattack{} can significantly degrade execution performance of blockchains, under a reasonable financial cost.
Understanding and mitigating the threats brought by \dosattack{} are crucial for ensuring the stability and resilience of blockchain ecosystems. 
Our contributions shed light on the importance of securing the state storage in blockchain, encouraging further research and development to safeguard against \dosattack{} and similar attacks.

\noindent
\textbf{Acknowledgements}
The authors thank anonymous reviewers for their constructive comments.
This work is partly supported by Hong Kong RGC Projects (No. PolyU15222320), National Natural Science Foundation of China under Grant No. 62332004, U22B2029 and U19A2066, the Natural Science
Foundation of Sichuan Province under Grant 2022NSFSC0871,
and the Financial Support for Outstanding
Talents Training Fund in Shenzhen.

\bibliographystyle{IEEEtran}
\bibliography{ref}

\appendices %

\ignore{
\section{Background of blockchain basic concepts}
\label{sec_background_blockchain}

We utilize the implementation of Ethereum~\cite{wood2014ethereum} to introduce the basic knowledge of blockchain. \chain{} is a widely used blockchain platform. According to its specifications~\cite{wood2014ethereum}, the basic structure of its data is the block, which comprises the header and body of the block. The block header involves a reference to the preceding block and the information used for state validation~\cite{wood2014ethereum}, while the block body contains a sequence of transactions. These transactions are signed by blockchain users to transfer funds and communicate with smart contracts.

There are two types of accounts in Ethereum, i.e., contract accounts (CA) and externally owned accounts (EOA). An EOA is controlled by a private key held by a user, while a CA holds the pre-defined logic and persistent variables. A contract is a Turing-complete automation program on Ethereum. The execution of contracts is facilitated by the Ethereum Virtual Machine (EVM), which is an underlying component of Ethereum supporting a set of instructions~\cite{wood2014ethereum}. 
Ethereum's native cryptocurrency is Ether. 

The gas mechanism~\cite{gas_intro} establishes the cost associated with users utilizing the blockchain's resources. 
For example, each operation in Ethereum, which modifies the state, will introduce the cost of gas,
e.g.,
executing contracts and transferring funds (e.g., Ether) will consume gas~\cite{wood2014ethereum}.
}

\ignore{
\begin{table*}[]
\centering
\caption{Time cost for $\dosattack{}$ to collide distinct lengths of desired prefix for an \dosformula{indexing}}
\label{tab_preiamges_full}
\resizebox{0.99\linewidth}{!}{
\begin{tabular}{@{}l|l|l|c@{}}
\toprule
\begin{tabular}[c]{@{}l@{}}Digits\end{tabular} & Crafted data (hex encoding) & $\dosformula{indexing}$ (hex encoding) & Time cost  \\
\midrule \midrule
1 & 0x51b0e4b84afc9c7e935fd1c54409abda46ffff07 & 0x$\underline{1}$09999afd60b733da226a060260c2d9f165f0f33516c5a3230d2b9538ae197e7 & $\dosformula{$<$1s}$  \\
2 & 0x7c0caee5b72d0c71a090c6f02522e89acfffff07 & 0x$\underline{11}$fb9e6a64c5a7c23fb27d08e3d74ea1018fcb0c60d2010cca6c6654dd95e4b8 & $\dosformula{$<$1s}$  \\
3 & 0x8f5ea3c9db43de4143e5717f44dcb43e05d0fe07 & 0x$\underline{111}$0dc62b86ce4609e860381909da5480d46b2e90ea19c5afac287be805c234b & $\dosformula{$<$1s}$  \\
4 & 0xbd6f8cba28b4a0218d0aedbc820a27248ee4fe07 & 0x$\underline{1111}$65e10752633a1ab85c219c618d6c6e6259fdb7c8d2397df9cb72d16e4149 & $\dosformula{$<$1s}$  \\
5 & 0xfccedcfd14858e8b1baf9a497e99af468012b507 & 0x$\underline{11111}$0e0c5d11a713c428c42a03a5a7c55d66c0e61158ef13a63776b94d384d0 & $\dosformula{$<$1s}$  \\
6 & 0x58b91f9cb0ffacae5d95c9e80c373d264993cc06 & 0x$\underline{111111}$078c719cdc5abc2195b645a72ba7dd4d15b12ab9cce3361466c402df69 & $\dosformula{$<$1s}$ \\
7 & 0x89f25e63c12c48a95c22cd4b19585f337a805f06 & 0x$\underline{1111111}$06b6090ca5f7027a7539dc73173e26a35b28645b47d4878db6bbddd62 & $\dosformula{$<$1s}$ \\
8 & 0xa0f0722109f07edd76cc1d2b29cfbc0122ca2b06 & 0x$\underline{11111111}$ce35790ede4c97cc847e55c91c0b3063f5cb56ab6ab93ee76381fa6a & $\dosformula{$<$1s}$ \\
9 & 0x97637e992f835689667a48a0731ce1ebb44dc006 & 0x$\underline{111111111}$0b3bf4ed6dc409fb20328970a0f23dac93761a4347fcd4c84dfe8cc & $\dosformula{53s}$  \\
10 & 0x2f1033b78f8fb3c04259202793d2d89169326d02 & 0x$\underline{1111111111}$ce8bad4529bfef324c88454fe4e72c3cd3974c0249c9adc764802a & $\dosformula{21.68m}$  \\
11 & 0x267a239f1986295e996358a79f57b473ae264d05 & 0x$\underline{11111111111}$00822f67e0319be36eb814ade0ca60c65c62b41641e889eb48ad8 & $\dosformula{2.8h}$\\
12 & 0xd4dfd776a81fcdfa2d601f1efa31a2ad8c21fe06 & 0x$\underline{111111111111}$834eea3006374356f398b29f9b709272533e759348f0bb07aa11 & $\dosformula{12.57h}$ \\
13 & 0xdf04b72b67666a59ff30c06dd079f1850b36ba04 & 0x$\underline{1111111111111}$ca536d3de683a3ab986f631ee733132457eccc0d9a011aa9e55 & $\dosformula{24.58h}$ \\ \bottomrule
\end{tabular}
}
\vspace{3pt}
\end{table*}
}

\section{Hash collision of \dosattack{}}
\subsection{The method of hash collision}
\label{sec_collsion}
In the following, we introduce how we
collide the prefix of a target keccak256 hash value (i.e., \dosformula{indexing} of a leaf node). We assume that adversaries do not have any advanced knowledge of cryptography, and they apply the most primitive brute force hash collision strategy~\cite{joux2007hash}. 
Brute force hash collision refers to exhausting the results of all possible keccak256 hash calculations, until a keccak256 hash value is crafted and satisfies adversaries' requirements. 
Concretely, the adversaries first pick a keccak256 hash value, and then
choose different inputs \dosformula{x$_\dossmallformula{i}$}, where \dosformula{i} $\in \mathbb{N}^{*}$.
After that, the adversaries check whether the prefix of the hash result of \dosformula{x$_\dossmallformula{i}$} matches the prefix of the picked keccak256 hash value. 

\noindent
\textbf{Multi-target hash collision.}
Similarly,
it does not rely on any advanced knowledge of cryptography. 
To conduct a multi-target hash collision, adversaries first pick $\phi$ target keccak256 hash values, and then choose different inputs \dosformula{x$_\dossmallformula{i}$}, where \dosformula{i} $\in \mathbb{N}^{*}$.
After that, the adversaries check whether the prefix of the hash result of \dosformula{x$_\dossmallformula{i}$} matches the prefix of one keccak256 hash value of the $\phi$  picked keccak256 hash values.
The adversaries finalize the multi-target hash collision until all the $\phi$  picked keccak256 hash values are matched.

\ignore{
\subsection{Time cost of hash collision}
\label{sec_app_hashtimecost}

Table~\ref{tab_preiamges_full} displays time cost of \dosattack{} to collide distinct lengths (1 - 13 nibbles) of desired prefix for an \dosformula{indexing}.
}

\ignore{
\begin{table}[!b]
\caption{Cost optimization based on active accounts}
\label{tab_optimized_cost}

\resizebox{0.99\linewidth}{!}{
\begin{tabular}{@{}l|c|c|c@{}}
\toprule
Count & Retained impact
& Retained cost
& \makecell{Optimized \dosformula{G}$_{\dossmallformula{gas}}$ (USD)}
\\
 \midrule \midrule
 1&100.00\% & 100.00\% & 11,808,917.46 \\
 2& 71.67\% & 19.60\% & 2,314,547.82\\
 4 & 58.16\% & 5.76\% & 680,193.64 \\
 6 & 54.66\% & 3.50\% & 413,312.11\\
\bottomrule
\end{tabular}
}
\end{table}
}

\section{Estimating the cost of \dosattack{}}
In this section, we will elaborate on how to estimate the cost of \dosattack{} before launching the attack.
As mentioned in~\S\ref{sec_evaluate_mainnet}, the attack impact can also be estimated.
Hence, the adversaries of \dosattack{} can strategically determine the trade-off of their actual attack, by firstly estimating the cost and attack impact of launching \dosattack{}. 
\subsection{The cost of units of gas}
\label{sec_gasguage}
\dosattack{} needs to cost 289,647,227,381 units of gas during the attack (\S\ref{sec_evaluate_eco}). 
Here we elaborate on how to estimate the cost of units of gas by Eq.~\ref{eq_gas0} - \ref{eq_gas2}. 
Besides, we reuse the same symbols in \S \ref{sec_evaluate_eco} and \S \ref{sec_evaluate_mainnet}. 
Eq.~\ref{eq_gas0} indicates that the units of gas consist of two parts, i.e., the units of gas consumed for inserting leaf nodes in State Trie \colorboxred{3} (\dosformula{Gas}${_\dossmallformula{StateTrie}}$) and the units of gas consumed for inserting leaf nodes in Storage Tries \colorboxred{4} (\dosformula{Gas}${_\dossmallformula{StorageTries}}$). 
According to~\S\ref{sec_design}, the units of gas consumed for inserting single leaf node in State Trie \colorboxred{3} and Storage Tries \colorboxred{4} are different.
Eq.~\ref{eq_gas1} and Eq.~\ref{eq_gas2} adopt the same method to calculate the units of gas for \dosformula{Gas}${_\dossmallformula{StateTrie}}$ and \dosformula{Gas}${_\dossmallformula{StorageTries}}$. 
For Eq.~\ref{eq_gas1} as an example, \dosformula{Gas}${_\dossmallformula{StateTrie}}$ is obtained as the product of the total number of leaf nodes that \dosattack{} needs to insert (i.e., \dosformula{Num}$_{\dossmallformula{StateTrie}}^{\dossmallformula{insert}}$) and the units of gas required to insert single leaf node in State Trie \colorboxred{3} (i.e., 21,000).
The units of gas for inserting a single leaf node in State Trie \colorboxred{3} is 21,000 units of gas, because we insert a leaf node in State Trie \colorboxred{3} by sending 1 wei to a target account (\S\ref{sec_design})~\cite{wood2014ethereum}.

\begin{small}
\useshortskip
\vspace{3pt}
\begin{gather}
\label{eq_gas0}
\dosformula{Units}_{\dossmallformula{gas}}\dosformula{ = }
\dosformula{Gas}_{\dossmallformula{StateTrie}}\dosformula{+}\dosformula{Gas}_{\dossmallformula{StorageTries}} \\
\label{eq_gas1}
\dosformula{Gas}_{\dossmallformula{StateTrie}}\dosformula{ = }
\dosformula{Num}_{\dossmallformula{StateTrie}}^{\dossmallformula{insert}} \times \dosformula{21000} \\
\label{eq_gas2}
\dosformula{Gas}_{\dossmallformula{StorageTries}}\dosformula{ = }
\dosformula{Num}_{\dossmallformula{StorageTries}}^{\dossmallformula{insert}} \times \dosformula{Cost}_{\dossmallformula{StorageTries}}
\end{gather}
\end{small}

The only difference between Eq.~\ref{eq_gas1}, and Eq.~\ref{eq_gas2} is the cost of inserting a leaf node.
This is because, according to \S\ref{sec_design}, we insert a leaf node in Storage Tries \colorboxred{4} by invoking a function of a target contract, hence, we denote the units of gas cost by inserting a leaf node in Storage Tries \colorboxred{4} as \dosformula{Cost}$_{\dossmallformula{StorageTries}}$,
e.g., if we invoke \doscode{transfer()} (\S\ref{sec_design}), it costs 44258 units of gas for inserting a leaf node in \mpt{}.

\begin{table}[!b]
\centering
\vspace{2pt}
\caption{Response status of seven mainstream blockchains}
  \label{tab_responce}
\resizebox{0.99\linewidth}{!}{
\begin{tabular}{@{}l|l@{}}
\toprule
\begin{tabular}[c]{@{}l@{}}Blockchain\end{tabular} & Response status\\
 \midrule \midrule 
Ethereum  & Accepted the vulnerabilities. \\ \midrule 
\begin{tabular}[c]{@{}l@{}}
Binance Smart Chain
\end{tabular}   & \begin{tabular}[c]{@{}l@{}}Accepted with the vulnerabilities, and rewarded\\ us with bug bounty.\end{tabular} \\ \midrule
Heco      & No reponse     \\ \midrule
Polygon   & Accepted the vulnerabilities. \\ \midrule
Optimism  & Accepted the vulnerabilities. \\ \midrule
Avalanche & \begin{tabular}[c]{@{}l@{}}Accepted with the vulnerabilities, and rewarded\\ us with bug bounty.\end{tabular}\\ \midrule
\begin{tabular}[c]{@{}l@{}}
Ethereum Classic
\end{tabular}       & Accepted the vulnerabilities.  \\
\bottomrule
\end{tabular}
}
\end{table}

\subsection{Cost of computing resources}
\label{sec_app_hashtimecost}
According to \S\ref{sec_evaluate_eco}, to deepen
all 2,754,284 leaf nodes,
we need 33 RTX3080 GPUs for the rental of 12 hours, i.e., 39.6 USD. Here we detail how to estimate the amount of GPUs and the required time. 
Besides, we reuse the same symbols in \S \ref{sec_design} and \S \ref{sec_evaluate_cap}. 

First, we need to confirm the number of leaf nodes (i.e., $\phi$) whose \dosformula{indexing} are required to be collided by \dosattack{}. 
We can then estimate the minimal number of hash collisions that are required to be satisfied for \dosattack{}.
According to \S\ref{sec_evaluate_mainnet}, and \textbf{S-3} in \S\ref{sec_design}, the minimal number of hash collisions equals to $\phi \times \frac{\dosformula{d}_{\dossmallformula{nurgle}}\dosformula{-}\dosformula{d}_{\dossmallformula{base}}}{\dosformula{2}}$.
We further adopt the multi-target hash collision strategy (\S\ref{sec_design}) on multiple GPUs to collide the \dosformula{indexing} of all $\phi \times \frac{\dosformula{d}_{\dossmallformula{nurgle}}\dosformula{-}\dosformula{d}_{\dossmallformula{base}}}{\dosformula{2}}$ targets.
As mentioned in~\S\ref{sec_design}, we parallelized conduct the keccak256 hash calculations, and we denoted \dosformula{GPU}$_\dossmallformula{time}$ as the time (hours) for single GPU to finish the multi-target hash collision.
We denote $\theta$ as the calculation counts for colliding one \dosformula{indexing}, and \dosformula{P} as the calculation counts that single GPU can finish in one hour.
Hence, we can estimate \dosformula{GPU}$_\dossmallformula{time}$ by Eq.~\ref{eq_index_colliding}.
For example, 
according to~\S\ref{sec_evaluate_cap}, $\frac{\theta}{\dosformula{P}}$ equals to 24.58 hours to deepen a leaf node to the layer 15.
According to \S\ref{sec_evaluate_eco}, there are 2,754,284 leaf nodes to be collided, i.e., $\phi$ equals to 2,754,284.
According to \S\ref{sec_evaluate_mainnet}, $\dosformula{d}{_{\dossmallformula{nurgle}}}$ and $\dosformula{d}{_{\dossmallformula{base}}}$ equal to 15 and 9.5, respectively.
Hence, $\dosformula{ln(} \phi \times \frac{\dosformula{d}_{\dossmallformula{nurgle}}\dosformula{-}\dosformula{d}_{\dossmallformula{base}}}{\dosformula{2}} \dosformula{)}$ equals to 15.84.
As a result, \dosformula{GPU}$_\dossmallformula{time}$ equals to 389.35 hours, i.e., single GPU needs 389.35 hours to collide the 2,754,284 leaf nodes.
In other words, it costs about 33 GPUs to collide the 2,754,284 leaf nodes in 12 hours (389.35/33).

\begin{small}
\useshortskip
\vspace{3pt}
\begin{gather}
\label{eq_index_colliding}
\dosformula{GPU}_\dosformula{time} \dosformula{=} 
\frac{\theta}{\dosformula{P}} \times \dosformula{ln(} \phi \times \frac{\dosformula{d}_{\dossmallformula{nurgle}}\dosformula{-}\dosformula{d}_{\dossmallformula{base}}}{\dosformula{2}} \dosformula{)}
\end{gather}
\end{small}

\ignore{
Concretely, we assign actual values to the Eq.\ref{eq_index_factors} to \ref{eq_index_colliding}. 
\dosformula{GPU}${_\dossmallformula{num}}$. Through Eq.\ref{eq_index_factors} and Eq.~\ref{eq_multitarget}, $\phi$ equals 15.84 (\dosformula{=}$\ln \dosformula{(}\frac{(2041719+712565)\times(15-9.5)}{2} \dosformula{)}$). Note that we assume that all \dosformula{indexing} is at the layer 13 to get the maximum number of GPUs. Next, we combine Table.\ref{tab_preiamges_full} and Eq.\ref{eq_index_colliding} to get \dosformula{T}$_{indexing}$ as 12 hours (\dosformula{=}$ \frac{24.58}{33}\times15.84$) and \dosformula{T}$_{indexing}$). This means that we can use 33 GPUs to calculate the \dosformula{indexing} required to attack 10,000 blocks (about 33 hours) within half a day.
}

\ignore{
\section{Optimized cost of \dosattack{} on Ethereum}
The cost of \dosattack{} can be significantly decreased when \dosattack{} only deepens the leaf nodes associated with active accounts in \mpt{}.
In such cases, most of the impact of the original attack can be retained.
In Table~\ref{tab_optimized_cost}, we provide the optimized cost of \dosattack{} on Ethereum based on different optimization parameters, and the proportion of the impact of the original attack that retains.
}

\ignore{
\subsection{Exploration the optimal strategy of \dosattack{}}
\label{sec_alg}

\begin{algorithm}[!b]
\caption{Optimal strategies of \dosattack{}}
\label{alg_maxattack}
\footnotesize{
  \KwOut{\dosformula{Strategies}, the optimal strategies to launch \dosattack{}}

\SetArgSty{text} 
\dosformula{Strategies}$_{\dossmallformula{Optimal}}$ $\leftarrow$ $\emptyset$\\
\dosformula{Blockchains} $\leftarrow$ Identify target victim blockchain platforms \\
\dosformula{Data}$_{\dossmallformula{Blockchains}}$ $\leftarrow$ Collect data from victim blockchains \\
\dosformula{GPUs} $\leftarrow$ Assess computing resources of adversaries\\
\dosformula{Height} $\leftarrow$ Decide the block height to launch \dosattack{}\\
\dosformula{Cost}$_{\dossmallformula{Upper}}$ $\leftarrow$ The upper bound of acceptable cost\\

\KwRet \dosformula{Strategies}
}
\end{algorithm}

\begin{algorithm}
\caption{The attack algorithm of \dosattack{}}
\label{alg_attack}
\begin{algorithmic}
\STATE{ \dosformula{$Blockchain \leftarrow$}  Identify target Blockchain of \dosattack{}}
\STATE{\dosformula{$Data \leftarrow$} Collect blockchain data}
\STATE{\dosformula{$GPUs \leftarrow$} Assess the resources of adversary}
\STATE{ \dosformula{$Height \leftarrow$} Choose the block height to attack}
\STATE{\dosformula{$Strategies = Nurgle(Data, GPUs, Height)$} Make attack strategies}
\STATE{\dosformula{$Impact = MAX(Blockchain, Strategies)$} Mount the attack}
\end{algorithmic}
\end{algorithm}

\noindent\textbf{The affecting factors of \dosattack{}.} i) Target Blockchain. Coin prices and gas fees vary from blockchain to blockchain, and this difference can greatly affect the cost of an attack. ii) The data on Blockchain. Data on the blockchain, such as the depth of the MPT or the frequency of access to the account, can determine the consequences of \dosattack{}. iii) Computing resources. The computational resources at the disposal of the adversary can determine the feasibility of \dosattack{}, i.e., how many intermediate nodes to increase. we depict a general attack algorithm considering such factors in Algorithm \ref{alg_attack}.

We provide a general attack algorithm for \dosattack{} in Algorithm. \ref{alg_attack}. The adversary first selects one blockchain as the attack target, then collects blockchain data (such as MPT node information), then evaluates the computing resources needed for the attack. Based on the above data and information, adversaries can make various attack strategies and choose the optimal strategy to launch an attack.
}

\ignore{
\section{The impact and cost of on-demand attacks}
\label{sec_on_demand}
In the following, we elaborate on how we evaluate the two on-demand attacks and the detailed experimental results.
 
\noindent
$\bullet$
In the first attack, the adversary slows down all users of an AMM (Automated Market Maker) contract~\cite{xu2021sok}, where users can exchange two specific tokens with the AMM contract. Since the token balances of the AMM contract will update when users interact with it, the adversary can deepen the two leaf nodes storing the token balances of the AMM contract for the two tokens to slow down all user transactions. The cost of transaction fees is 4.06 USD on Ethereum, and the cost of computing resources is 4.32 USD.

\noindent
$\bullet$
In the second attack, the adversary delays all users of a token contract. Since, for each transaction, only the token balances of the users in the current transaction will update, the adversary needs to deepen all the leaf nodes storing all users' token balances to slow down all user transactions. In our experiments, the token contract has 10,000 users. Hence, the cost of transaction fees is 20,154.95 USD on Ethereum, and the cost of computing resources is 23.56 USD.
}

\ignore{
\section{Three feasible mitigations against \dosattack{}}
\label{sec_mitigations}
In the following, we detail the three feasible mitigations and discuss their advantages and disadvantages.

\noindent
i) \textbf{Verkle tree}~\cite{kuszmaul2019verkle}. Verkle tree mitigates the impact of \dosattack{} by indexing fewer nodes in its tree structure and adopting a swifter authenticated method~\cite{kuszmaul2019verkle}.
First, compared with the prefix tree in MPT structure, the structure of verkle tree is designed to be flatter, which compresses the distance from the leaf node to the root node. The verkle tree alleviates the volume of nodes to be updated and verified associated with \textbf{OP1} and \textbf{OP2} (\S\ref{sec_observation_stateop}). 
Second, verkle tree adopts a polynomial commitment scheme~\cite{boneh2021halo} rather than the hash-style vector commitment of MPT structure~\cite{wood2014ethereum}.
During the verification for the latest state, the polynomial commitment reduces the number of nodes to be verified~\cite{boneh2021halo}. 
Verkle tree has the advantage of consuming fewer resources when accessing a leaf node by involving fewer intermediate nodes. However, the advantage comes at the cost of increased space wastage in Verkle tree, as each node consumes more space. For instance, a branch node in Verkle tree, which includes 256 pointers, occupies nearly 16 times more space than a branch node in \mpt{}, which only has 16 pointers.

\noindent
ii) \textbf{Trimming historical state (EIP-4444)}~\cite{eip4444}. EIP-4444 eases \dosattack{}'s impact by periodically reducing \mpt{}'s size. 
Specifically, \mpt{} will undergo pruning by retaining the state of blockchain for nearly one year. Ethanos~\cite{ethanos} has carried out its implementation in a similar manner. 
Trimming historical state has the advantage of reducing the storage space required by a blockchain node by eliminating a portion of its historical state. However, the mitigation comes with the disadvantage that the blockchain node cannot conduct complete security verification of the whole blockchain state due to the elimination of the historical state.

\noindent
iii) \textbf{New patch for gas mechanism}. We plan to propose a new Ethereum improvement proposal (EIP) for the gas mechanism to defend against \dosattack{}. In our EIP, gas fee of transactions will be equivalent to actual consumed resources, e.g., considering the consumed resources for modifying nodes in \mpt{}.
The mitigation has the advantage of implementing a fairer gas mechanism. The adjusted mechanism ensures that the gas fee for transactions corresponds directly to the actual resources consumed when modifying nodes in the MPT. However, the mitigation can introduce new attack vectors. For instance, an adversary can exploit the adjusted mechanism by proliferating intermediate nodes to popular contracts, thereby increasing gas fees for all users of the contracts.
Hence, it is crucial to examine the implementation of the third mitigation thoroughly.
}

\section{Response status of blockchains}
\label{sec_app_status}

We have reported the vulnerabilities exploited by \dosattack{} to seven mainstream blockchain platforms, including Ethereum, Binance Smart Chain, Heco, Polygon, Optimism, Avalanche, and Ethereum Classic.
In Table~\ref{tab_responce}, we summarize their latest responses for the corresponding vulnerabilities.
Specifically, six blockchains of them, i.e., Ethereum, Binance Smart Chain, Polygon, Optimism, Avalanche, and Ethereum Classic have accepted the vulnerabilities and were exploring
appropriate countermeasures against \dosattack{}. Especially, we have received thousands of USD bounty from Binance Smart Chain and Avalanche.
Furthermore, before the publication, we reported the vulnerabilities to all other affected blockchains.
Currently, an additional 16 teams from the newly reported blockchain teams have responded with positive acknowledgments.
Comprehensive details regarding their responses can be found in our repository at \url{https://github.com/hzysvilla/Nurgle_Oakland24}, and we will keep updating their feedback.

\subsection{Vulnerable blockchains}
\label{sec_scope}
\dosattack{} threatens
588 blockchains compatible with Ethereum and 153 blockchains compatible with Polkadot. We enumerate the blockchains in \url{https://github.com/hzysvilla/Nurgle_Oakland24}. 

\newpage %

\section{Meta-Review}

The following meta-review was prepared by the program committee for the 2024
IEEE Symposium on Security and Privacy (S\&P) as part of the review process as
detailed in the call for papers.

\subsection{Summary}
This paper presents a denial of service attack targeting the blockchains using the Merkle Patricia Trie (MPT) structure. The attack inserts new intermediate nodes to force the network to use more resources to keep the network state. The attack directly applies to popular blockchains such as Ethereum. 

\subsection{Scientific Contributions}
\begin{itemize}[leftmargin=*,topsep=1pt]
\item Identifies an Impactful Vulnerability
\end{itemize}

\subsection{Reasons for Acceptance}
\begin{enumerate}[leftmargin=*,topsep=1pt]
\item The paper finds a novel denial service attack on popular blockchains.
\item The attack is tested on the testnet and acknowledged by the blockchain labs.
\item The paper demonstrates the attack is viable.
\item The paper is well-organized and easy to follow.
\end{enumerate}

\ignore{
\subsection{Noteworthy Concerns} %
\begin{enumerate}[leftmargin=*,topsep=1pt] %
\item Practical Impact of the Attack. The reviewers have raised questions about the practical impact of the attack. A more comprehensive discussion is needed, specifically, on the end-to-end performance degradation to blockchains, its impact beyond the current implementation of Ethereum blockchains, and the performance overhead of the proposed system.
\item Proposed Mitigations. The reviewers also suggested that the paper should discuss the pros and cons (\& efficacy) of the proposed mitigations in more detail.
\item Vulnerability Disclosure. Finally, the reviewers also expressed their concerns about the vulnerability disclosure process in the paper. The paper should follow the best practices of S\&P 2024 as promised in the rebuttal and inform all affected blockchains prior to the publication of the paper.
\end{enumerate}
}

\ignore{
\section{Response to the Meta-Review} %

\noindent
Dear Reviewers,

Thank you so much for providing constructive comments and a
detailed list of concerns. 
The revision changes of our paper have incorporated the suggestions outlined in the rebuttal process as well as the concerns raised in the meta-review.
Please find two attached PDF files, and our revision summary below. 

\noindent
(1) File revised\_paper.pdf. This is the revised version of our paper. %

\noindent
(2) File diff.pdf: This is the diff file that shows the changes made in our revision. Please note that the text marked with blue is the newly added text to address the concerns, while the removed text is marked with red.

\subsection{Revision summary}

\noindent
\textbf{Concern 1. Practical Impact of the Attack}

We explored the practical impact of \dosattack{} as follows.

\begin{itemize}[leftmargin=*,topsep=1pt]
    \item In \S\ref{sec_dis_impact}, we elaborate on how \dosattack{} can threaten blockchains from seven aspects.
    \item In \S\ref{sec_evaluate_testnet}, we evaluate the end-to-end performance degradation to Ethereum and BSC testnets brought by \dosattack{}, i.e., 10.7\% and 12.4\%, respectively.
    \item In \S\ref{sec_on_demand}, we explore the feasibility of on-demand attacks of \dosattack{}, and evaluate the cost and impact of on-demand attacks by two cases.
    \item In \S\ref{sec_dis_impact} and \S\ref{sec_vul_disclose}, we discuss how \dosattack{} can threaten layer 2 rollups and blockchains compatible with the Polkadot ecosystem, respectively, which are both beyond the current implementation of Ethereum blockchains.  
\end{itemize}

\noindent
\textbf{Concern 2. Proposed Mitigations}

In \S\ref{sec_mitigations}, in addition to illustrating the implementation and defense mechanisms against \dosattack{} for each mitigation, we also detail the advantages and disadvantages of each mitigation.

\noindent
\textbf{Concern 3. Vulnerability Disclosure}

At the time of writing, we have reported the vulnerabilities to all affected blockchains (Appendix \ref{sec_scope}).
To inform the vulnerable blockchains, we have sent our revised paper via email. In cases where we cannot find the corresponding email addresses, we have initiated conversations with blockchain teams through social media platforms such as Twitter.
Currently, in addition to the six blockchain teams that acknowledged our vulnerabilities as outlined in our original paper, an additional 16 teams from the newly reported blockchain teams have responded with positive acknowledgments. 
We have accordingly updated the vulnerability disclosure procedure, and responses from vulnerable blockchain teams in \S\ref{sec_vul_disclose} and Appendix~\ref{sec_app_status}.
Comprehensive details
regarding their responses can be found in our repository at \url{https://github.com/hzysvilla/Nurgle_Oakland24}, and we will keep updating their feedback.

\noindent
\textbf{Suggestion 1. Cost of creating spam transactions compared to the values in Table 3}

We conducted the baseline comparison in \S\ref{sec_evaluate_eco}. Experimental results show that creating spam transactions costed 3,826,037.45 USD on Ethereum, which is 9.25 times higher than the value (i.e., 413,345.11 USD) in Table 3.

\noindent
\textbf{Suggestion 2. The introduction is not focused}

We streamlined the technical details in the introduction section to make the motivation more clear.
Besides, we relocated contents pertaining to the background and related work sections from the introduction to make the introduction more self-contained.

\noindent
\textbf{Suggestion 3. Include several historical state exploits as part of the motivation for this paper}

We clarify that the design of \dosattack{} is motivated by existing studies that exploit historical state in \S\ref{sec_related}.

\noindent
\textbf{Suggestion 4. Adjusting the gas fee per cost unit as a mitigation can induce new attack vectors.}

We clarify that the patch for gas mechanism can introduce new attack vectors, and it is crucial to examine the implementation of the corresponding mitigation thoroughly in \S\ref{sec_mitigations}.

Please kindly let us know if the revised paper is satisfactory. Thank you so much for taking your precious time to help improve our paper, and we would be very thankful to hear your further feedback.

Sincerely,

Authors of Submission \#238
}

\end{document}